\documentclass[11pt]{article}

\usepackage[margin=1in]{geometry}

\usepackage[T1]{fontenc}
\usepackage[utf8]{inputenc}
\usepackage{lmodern}

\usepackage{amsmath,amssymb,amsthm,mathtools}

\usepackage[numbers,sort&compress]{natbib}

\usepackage[dvipsnames]{xcolor}
\usepackage[colorlinks=true,linkcolor=blue,citecolor=blue,urlcolor=blue]{hyperref}

\numberwithin{equation}{section}

\theoremstyle{plain}
\newtheorem{theorem}{Theorem}
\newtheorem{lemma}[theorem]{Lemma}
\newtheorem{corollary}[theorem]{Corollary}

\theoremstyle{definition}

\newcommand{\disc}{\mathrm{disc}}
\newcommand{\herdisc}{\mathrm{herdisc}}
\newcommand{\wdisc}{\mathrm{wdisc}}
\newcommand{\detlb}{\mathrm{detlb}}
\newcommand{\la}{\langle}
\newcommand{\ra}{\rangle}

\definecolor{darkgreen}{RGB}{0,100,0}

\title{Approximation algorithms via discrepancy with hard constraints}
\author{
Mikhail Cherniavskii, Andrey Kupavskii, Stanislav Moiseev, Alexander Shekhovtsov
}
\date{}

\begin{document}

\maketitle

\begin{abstract}
In this paper, we present a discrepancy-based approach to cost-preserving rounding under hard constraints.
Starting from a feasible fractional solution, our rounding procedure preserves the designated hard constraints, never increases the linear objective, and bounds the error in the remaining constraints by the discrepancy of a suitable column-difference matrix. 

Our second discrepancy-theoretic result bounds the discrepancy of a matrix in terms of a novel column parameter, called {\em mountain distance}. This notion captures ordered structural
features of a matrix which are invisible to the classical Beck--Fiala and
Komlós column-norm parameters. 

We show a variety of applications of these two results. In particular, we obtain new results for min-cost unsplittable
flows, bounded-turn routing in grids and tori, rainbow prefix hypergraph
decomposition, and \(2\)-multicast scheduling.
\end{abstract}

\section{Introduction}
Discrepancy theory studies how to round or color a set of objects while keeping
 constraints nearly balanced. For a matrix \(A\in\mathbb R^{m\times n}\), its
discrepancy is
\[
    \disc(A)
    :=
    \min_{\chi\in\{-1,1\}^n}
    \|A\chi\|_\infty .
\]
For matrices that arise from set systems/hypergraphs, 
\(\disc(A)\) measures the imbalance for red-blue colorings. In combinatorial optimization, the relevance of discrepancy comes from its
connection to rounding linear programs. Lov{\'a}sz, Spencer and
Vesztergombi~\cite{lovasz1986discrepancy} showed that discrepancy parameters
control the additive error needed to round fractional solutions of linear
programs. Thus, discrepancy bounds can be viewed as rounding guarantees.

Discrepancy theory has a long history, including early work on hypergraph
decompositions and balancing
\cite{baranyai1974factrization,olson1978balancing,beck1981roth,beck1981integer,spencer1985six,gluskin1989extremal,banaszczyk1998balancing,
spencer2001discrepancy,banaszczyk2012series}. For a long time, many of these
results were mainly existential and therefore less directly usable in discrete
optimization. This changed with algorithmic discrepancy breakthroughs
\cite{bansal2010constructive,lovett2015constructive,
rothvoss2017constructive,bansal2019algorithm,bansal2026decoupling,guo2026vector}, which made
discrepancy a powerful tool for rounding linear relaxations.

Recent developments include an \(O((\log n)^{1/4})\) bound
for the Koml{\'o}s problem~\cite{ercan2026bound},
a proof of the Koml{\'o}s conjecture via directional total
variation~\cite{guo2026vector}, and an elementary
proof~\cite{karingula2026elementary}.
Guo, Fang, and Lu~\cite{guo2026vector} also gave a deterministic
polynomial-time algorithm for constant-discrepancy signing.

However, the standard discrepancy setup does not fully incorporate two features
that are central in approximation algorithms: objective functions and hard
feasibility constraints. In many rounding problems, some constraints may
tolerate small additive error, while assignment, partition, or conservation
constraints must be preserved exactly.
Related hard-constraint rounding questions appear in Beck--Fiala type results
\cite{biedl2002balanced,doerr2005roundings}, in entropy-method,
partial-coloring, and approximation-friendly discrepancy rounding
\cite{rothvoss2012entropy,bansal2014better,lovett2015constructive,bansal2017approximation,
bansal2019generalization}, in dependent-rounding
methods for bipartite and assignment-type problems
\cite{fiala1983algorithm,grigoriev2005unrelated,gandhi2006dependent,
kumar2009unified,rohwedder2025cost}, and in bin packing, where Hall-type
matching arguments preserve hard assignment constraints during rounding
\cite{eisenbrand2013bin,hoberg2017logarithmic}.

In this paper we develop a discrepancy-based rounding framework that
incorporates both hard constraints and a linear objective. We consider relaxations
of the form
\[
    \min\{\langle c,x\rangle : Ax=b_A,\ Hx=b_H,\ x\in[0,1]^n\}.
\]

Below, we state an informal version of our main rounding theorem.
Here, \(\mathcal Q(A)\) denotes a family of matrices whose
discrepancy can be bounded in terms of the hereditary discrepancy
of the original matrix \(A\). We assume that   \(H\in\{0,1\}^{k\times n}\) and
\(\|H_{\ast j}\|_1=1\) for every \(j\in[n]\).

\begin{theorem}\label{th:ROUGHdiscHardConstraintsTheorem}
Let \(A\), $H$ be the matrices as above, and let $c$ be the linear objective. 
Let $\beta > 0$, such that
$
    \max_{Q\in\mathcal{Q}(A)} \disc(Q) \leq \beta.
$

Then for any \(x\in[0,1]^n\) with \(Hx\in\mathbb{Z}^k\), one can efficiently find
\(y\in\{0,1\}^n\) such that
\[
    Hy=Hx,
    \qquad
    \langle c,y-x\rangle\le0,
\qquad 
    \|A(x-y)\|_\infty
    \le
    \beta+\frac{2}{n}\|A\|_\infty .
\]
\end{theorem}

The theorem reduces rounding to constructively bounding
\(\max_{Q\in\mathcal Q(A)}\disc(Q)\).
In many settings, the discrepancy estimates used for \(A\)
apply to this family with only a constant-factor loss.

The theorem also extends to block-diagonal hard-constraint
matrices whose blocks are totally unimodular and have bounded number of rows. For simplicity and clarity of exposition, we give
the statement and proof of this extension in
Appendix~\ref{app:block-TU-hard-constraints}.

We combine this framework with discrepancy bounds based on
two column parameters: 
the Koml{\'o}s \(\ell_2\) parameter, and the
\emph{mountain distance}
\[
    \mu(A):=
    \max_{j\in[n]}
    \sum_{i\in[m-1]}|A_{i+1,j}-A_{i,j}|.
\]
These yield new approximation guarantees for min-cost
unsplittable flow with bounded-hop paths, bounded-turn routing
in grids and tori, and rainbow prefix hypergraph decomposition.

Finally, we describe an application of the rainbow prefix hypergraph decomposition to a problem that arose in Telecom, to the so-called
 \(2\)-multicast scheduling problem. In fact, this problem was, to a large extent, responsible for the existence of the present paper.

\subsection{Notation}
Throughout the article we use the following notation. For a positive integer
\(n\), we write $[n]:=\{1,\ldots,n\}.$
All logarithms are base \(2\); in particular, \(\log n:=\log_2 n\). The notation
\(\widetilde O(\cdot)\), somewhat unconventionally, suppresses factors that are polynomial in logarithms of
logarithms ($\mathrm{poly} \log \log$) of the relevant parameters. For a matrix
\(A\in\mathbb R^{m\times n}\), we denote its \(j\)-th column by \(A_{\ast j}\)
and define $\|A\|_p:=\max_{j\in[n]}\|A_{\ast j}\|_p .$
If \(\mathcal F\) is a hypergraph and \(v\) is a vertex, then
\(d_{\mathcal F}(v)\) denotes the degree of \(v\) in \(\mathcal F\). Finally,
\(\binom{[n]}{k}\) denotes the family of all \(k\)-element subsets of \([n]\),
and \(\binom{[n]}{\le k}\) denotes the family of all subsets of \([n]\) of size
at most \(k\). For a directed or undirected graph or multigraph \(G\), we denote by \(V(G)\)
and \(E(G)\) its vertex set and edge set, respectively. 

\newpage
\subsection{Our results}

We prove the following results.

\begin{itemize}
    \item \textbf{Rounding with hard constraints.}
We provide a discrepancy-based rounding
method for linear programs with exact feasibility requirements. For programs of
the form $\min\{\langle c,x\rangle : Ax=b_A,\ Hx=b_H,\ x\in[0,1]^n\},$
where \(H\in\{0,1\}^{k\times n}\) satisfies
\(\|H_{\ast j}\|_1=1\) for every \(j\in[n]\), the theorem reduces rounding to
bounding the discrepancy of matrices formed from pairwise differences of
columns of \(A\). The resulting rounding preserves the hard constraints
encoded by \(H\) exactly, does not increase the objective value, and incurs only
a controlled additive error in the constraints encoded by \(A\).

    \item \textbf{A mountain-distance discrepancy bound.}
    We introduce the column parameter
    \[
        \mu(A)
        :=
        \max_{j\in[n]}
        \sum_{i\in[m-1]} |A_{i+1,j}-A_{i,j}|,
    \]
    which measures the total variation of each column along the row order. We
    prove a tight discrepancy bound
    \[ \disc(A)
        \le
        O(1) \cdot \mu(A) \cdot \log(2+\min\{m,n\})
        +
        2\|A\|_\infty .
    \]
    The parameter \(\mu(A)\) is designed for rounding settings in which the
constraints have underlying structure, as is the case in many resource-allocation and
scheduling problems. We apply it to study
min-cost bounded-turn unsplittable flow in grids and tori.

    \item \textbf{Min-cost unsplittable flow in general directed multigraphs.}
    We apply our rounding framework together with the constructive Koml{\'o}s
    discrepancy bound to min-cost unsplittable flow. Given requests
    $\{(s_i,t_i,\rho_i)\}_{i\in[k]},$
    where request \(i\) must be routed unsplittably from \(s_i\) to \(t_i\), and
    where admissible paths have at most \(D\) hops, we round any feasible LP
    solution to an integral routing. The algorithm selects one path per request,
    does not increase the LP cost, and, on every edge, the integral congestion
    exceeds the LP congestion by at most
    $O\!\left(\rho_{\max}\sqrt D\right),$
    where \(\rho_{\max}:=\max_i\rho_i\).

    \item \textbf{Min-cost bounded-turn unsplittable flow in grids and tori.}
    For \(d\)-dimensional grids and tori, we consider the structured case in
    which every admissible path has at most \(r\) turns. Here we apply the same
    hard-constraint rounding framework together with our mountain-distance
    discrepancy bound. This gives a cost-preserving rounding in which the
    additive congestion error on every edge is at most $O\!\left(
            \rho_{\max}(r+1)\log |V(G)|
        \right).$

    \item \textbf{Rainbow prefix hypergraph decomposition.}
    We apply our rounding framework together with the constructive
Koml{\'o}s bound 
to construct decompositions of regular colored hypergraphs. Let $\mathcal F=\mathcal C_1\sqcup\cdots\sqcup\mathcal C_n
    $
    be an \(N\)-regular colored multihypergraph, where each color class
    \(\mathcal C_i\) contains \(N\) edges of size \(k\). For \(N\) a power of
    two, we give a polynomial-time algorithm constructing prefixes
    $\varnothing=\mathcal F_0
        \subset \mathcal F_1
        \subset \cdots
        \subset \mathcal F_N=\mathcal F$
    such that each layer \(\mathcal F_i\setminus\mathcal F_{i-1}\) is rainbow
    and $d_{\mathcal F_i}(v)\in[i-\beta,i+\beta]
        \ \forall v,\ i,$
    with $\beta=O(\sqrt{k}\log N).$

    \item \textbf{\(2\)-multicast scheduling.}
    We apply the rainbow prefix decomposition theorem to a balanced
    \(2\)-multicast scheduling problem. There are \(n\) active in-ports,
    \(2n\) out-ports, and \(n\) rounds. Each in-port sends one message per
    round, each message has two destinations, and each out-port is a destination
    of exactly \(n\) message-copies. The goal is to order the messages so that
    one message from each in-port appears in every round, while minimizing the
    maximum number of undelivered copies kept in the buffer. We obtain a polynomial-time schedule with maximum buffer
    occupancy at most $O(n\log n).$
\end{itemize}

\subsection{Related discrepancy work}

The study of discrepancy has produced several fundamental results. The
Beck--Fiala theorem~\cite{beck1981integer} states that if every element appears
in at most \(t\) sets, or equivalently if every column of a \(0\)-\(1\) matrix
has \(\ell_1\)-norm at most \(t\), then the discrepancy is at most $2t$. Beck and
Fiala conjectured the optimal bound \(O(\sqrt t)\). A substantial
generalization is the Komlós conjecture, which asks for \(O(1)\) discrepancy
for every matrix whose columns have Euclidean norm at most \(1\); see also
Spencer's formulation of vector balancing~\cite{spencer1986balancing}. Another
landmark result is Spencer's ``six standard deviations suffice''
theorem~\cite{spencer1985six}, which gives discrepancy \(O(\sqrt m)\) for set
systems with  $m$ sets.

A different line of work uses convex geometry. Banaszczyk~\cite{banaszczyk1998balancing}
proved a powerful vector-balancing theorem based on Gaussian measure, implying
an \(O(\sqrt{\log n})\) bound for the Komlós problem and an
\(O(\sqrt{t\log n})\) bound for Beck--Fiala instances. Banaszczyk also obtained
strong bounds for prefix discrepancy~\cite{banaszczyk2012series}, where one
asks not only for the final signed sum to be small, but also for all prefix
sums to remain small. Bansal and Jiang~\cite{bansal2026decoupling}
obtained a \(\widetilde O(\log^{1/4} n)\) bound for the Komlós
problem and proved the Beck--Fiala conjecture when \(t\ge\log^2 n\).
More recently, Guo, Fang, and Lu~\cite{guo2026vector} gave a
deterministic polynomial-time algorithm achieving \(O(1)\) Komlós discrepancy and hence
\(O(\sqrt t)\) Beck--Fiala discrepancy for every \(t\).

Ordered discrepancy problems also show that row order can play an essential
role. For instance, Newman, Neiman and Nikolov~\cite{newman2012beck} disproved
Beck's three-permutations conjecture by constructing three permutations whose
associated set system has discrepancy \(\Omega(\log n)\). Such examples suggest
that even when the columns have a simple ordered structure, logarithmic losses
can be unavoidable.

Beyond existential coloring theorems, discrepancy has become an important
algorithmic rounding tool. Bansal~\cite{bansal2010constructive} gave a
constructive algorithm for discrepancy minimization, and Lovett and
Meka~\cite{lovett2015constructive} developed a constructive random-walk
approach to partial coloring. These methods have influenced approximation
algorithms and LP rounding. Rothvoss~\cite{rothvoss2016better} used
discrepancy-theoretic ideas to obtain improved bin-packing approximations, while
Bansal and Nagarajan~\cite{bansal2017approximation} developed
approximation-friendly discrepancy rounding, extending the Lovett--Meka approach to matroid settings. Prefix discrepancy has
also found applications in scheduling: Bansal, Rohwedder and
Svensson~\cite{bansal2022flow} related prefix Beck--Fiala bounds to the
integrality gaps of LP relaxations for flow-time scheduling.

\subsection{Structure of the paper}
The paper is organized as follows. In each section, we start by
formulating the problem or setting under consideration, present our results, then discuss related
work (if any), and finally present the proofs of the corresponding results.

Section~\ref{sec:DiscreapancyResults} contains our two discrepancy results: rounding with hard constraints and a mountain-distance bound. In Section~\ref{sec:min-cost-unsplittable-flow} we apply
these tools to the problem of determining min-cost unsplittable flow for general multigraphs, and then to bounded-turn unsplittable flow in grids and tori. In Section~\ref{sec:RainbowPrefix} we discuss rainbow prefix hypergraph decomposition.  Section~\ref{sec:twoMulticast} is on 2-multicast scheduling. Finally, in Section~\ref{sec:Conclusions} we give several open problems.

\section{Discrepancy results}\label{sec:DiscreapancyResults}

In this section we prove two discrepancy results used throughout the
paper.
First, we prove our main rounding tool, Theorem~\ref{th:discHardConstraints}.

The second result introduces   a column
parameter \(\mu(A)\) measuring the total variation of each column along the row order.
We prove a tight bound
\[
    \disc(A)
    \le
    O(1)\cdot \mu(A)\log(2+\min\{m,n\})
    +
    2\|A\|_\infty .
\]

\subsection{Discrepancy with hard constraints}

In many rounding problems, some constraints encode feasibility and must be
preserved exactly. In our setting, given \(A,H,c\) and a fractional point
\(x\in[0,1]^n\) with \(Hx\in\mathbb{Z}^k\), we want to find
\(y\in\{0,1\}^n\) such that
\[
    Hy=Hx,
    \qquad
    \langle c,y-x\rangle\le0,
\]
while keeping
\[
    \|A(x-y)\|_\infty
\]
small. Thus the rounding must preserve the hard constraints $H$ exactly and must not
increase the linear objective $\la c ,x\ra$.

Next, we state and prove our main rounding theorem.
Let $\mathcal{M}(n,m)$ be the set of all matchings on $[n]$ of size at most
$\lfloor\min\{ n/2 , m \}\rfloor$. For $A \in \mathbb{R}^{m \times n}$, define
\[
    \mathcal{Q}(A)
    :=
    \left\{
        Q : Q \in \mathbb{R}^{m \times |M|},
        Q_{\ast,\{j',j''\}}
        =
        A_{\ast j'}-A_{\ast j''}
        \text{ for all } \{j',j''\}\in M, 
        \;
        M\in\mathcal{M}(n,m)
    \right\}.
\]

An extension of this theorem to block-diagonal hard-constraint
matrices with totally unimodular blocks is stated and proved in
Appendix~\ref{app:block-TU-hard-constraints}, where we also
describe how the discrepancy bounds depend on the maximum
number of rows in a block.

\begin{theorem}
\label{th:discHardConstraints}
Let \(A\in\mathbb{R}^{m\times n}\),
\(H\in\{0,1\}^{k\times n}\) with \(\|H_{\ast j}\|_1=1\) for all
\(j\in[n]\), and \(c\in\mathbb{R}^n\). Let \(\beta\ge0\) be such that
$
    \max_{Q\in\mathcal{Q}(A)} \disc(Q)\le \beta,
$
and assume access to an oracle which, for every \(Q\in\mathcal{Q}(A)\), returns
a coloring witnessing \(\disc(Q)\le\beta\).
Then for any \(x\in[0,1]^n\) with \(Hx\in\mathbb{Z}^k\), one can find
\(y\in\{0,1\}^n\) such that
\[
    Hy=Hx,
    \qquad
    \langle c,y-x\rangle\le0,
\qquad 
    \|A(x-y)\|_\infty
    \le
    \beta+\frac{2}{n}\|A\|_\infty .
\]
Moreover, apart from oracle calls, the algorithm runs in polynomial time and
uses at most \(\lceil 2\log_2 n\rceil\) oracle calls.
\end{theorem}

\begin{proof}
Let \(x'\) be an optimal vertex of the polytope 
$
    \{z\in[0,1]^n:Az=Ax,\ Hz=Hx\}
$
with respect to the objective \(\langle c,z\rangle\). Then
\[
    Ax'=Ax,
    \qquad
    Hx'=Hx,
    \qquad
    \langle c,x'-x\rangle\le0 .
\]
Let
$
    F:=\{j:0<x'_j<1\},
$
fix all coordinates outside \(F\), and let \(I\) be the rows of \(H\) that meet
\(F\). If \(F=\emptyset\), take \(y=x'\). Otherwise, put
$
    n':=|F|,
    k':=|I|.
$
Since \(x'\) is a vertex of the polytope, the columns indexed by \(F\) in $\binom{A}{H_I}$
are linearly independent, hence \(n'\le \mathrm{rank \ of \  } \binom{A}{H_{I}}\leq  m+k'\). Also, every row in \(I\)
contains at least two coordinates from \(F\); otherwise the unique fractional
coordinate in that row would be forced to be integral. Since every column of
\(H\) contains exactly one \(1\), these supports are disjoint, so \(n'\ge2k'\).
Thus $2k'\le n'\le m+k'$ implies $k'\le m$ and we have  $n'\le2m$.

We next make one more simplification: instead of rounding \(x'\) directly, we
replace it by a nearby vector \(w\) whose coordinates have binary expansion only
up to level \(L\). This replacement preserves the hard constraints, does not
increase the objective, and creates additive error at most
\(\frac{2}{n}\|A\|_\infty\).

Set $L:=\lceil 2\log_2 n\rceil .$
For \(j\in F\), write
$
    x'_j=\sum_{s=0}^{\infty} b_j^s2^{-s}
$ for some binary $b_j^s$
and define
$
    x''_j:=\sum_{s=0}^{L} b_j^s2^{-s}.
$
For \(j\notin F\), set \(x''_j:=x'_j\). Then
$
    0\le x'_j-x''_j\le 2^{-L}\le n^{-2}.
$

Relabel the rows in \(I\) by \([k']\), and define
$
    N_\ell:=\{j\in F:H_{\ell j}=1\} 
$ for all $\ell\in[k']$.
The sets \(N_1,\ldots,N_{k'}\) partition \(F\). For each \(\ell\in[k']\), let
$
    \Delta_\ell:=\sum_{j\in N_\ell}(x'_j-x''_j).
$
Since \(\sum_{j\in N_\ell}x'_j\) is integral and \(x''\) has binary expansion up
to level \(L\), the number \(2^L\Delta_\ell\) is an integer.

Starting from \(x''\), construct \(w\) independently on each \(N_\ell\): add
total mass \(\Delta_\ell\) inside \(N_\ell\), in steps of \(2^{-L}\), greedily
to coordinates with smallest \(c_j\), never exceeding \(1\). Finally set
\(w_j:=x'_j\) for \(j\notin F\). Then each row sum is restored, so
$Hw=Hx'=Hx.$ The greedy
rule adds it to the cheapest coordinates, hence
$\langle c,w-x'\rangle\le0.$
Also,
$w_j=\sum_{s=0}^{L} q_j^s2^{-s}
    \text{ for some } q_j^s\in\{0,1\}.$
Finally,
\[
    \|x'-x''\|_1\le n'2^{-L}\le \frac1n,
    \qquad
    \|x''-w\|_1=\sum_{\ell=1}^{k'}\Delta_\ell\le n'2^{-L}\le \frac1n.
\]
Since \(Ax'=Ax\), we get
$
    \|A(x-w)\|_\infty
    =
    \|A(x'-w)\|_\infty
    \le
    \|A\|_\infty\|x'-w\|_1
    \le
    \frac{2}{n}\|A\|_\infty .
$

It remains to round \(w\) to an integral vector. Let \(r\ge1\) be the largest level with
\(q_j^r=1\) for some \(j\). For every \(\ell\in[k']\), the number of
\(j\in N_\ell\) with \(q_j^r=1\) is even, since
\(\sum_{j\in N_\ell}w_j\in\mathbb{Z}\). Pair these variables inside
each \(N_\ell\), and let \(M\) be the resulting matching.

For every pair \(e=\{j',j''\}\in M\), fix an arbitrary orientation and define
$Q_{\ast e}:=A_{\ast j'}-A_{\ast j''}.$
Since \(|M|\le n'/2\le m\), we have \(Q\in\mathcal{Q}(A)\). Hence the
discrepancy oracle gives \(\chi\in\{-1,1\}^{M}\) with $\|Q\chi\|_\infty\le\beta.$
Replacing \(\chi\) by \(-\chi\), if necessary, we may also assume
$\sum_{e=\{j',j''\}\in M}\chi_e(c_{j'}-c_{j''})\le0.$
Define \(w'\) by
\[
    w'_{j'}:=w_{j'}+2^{-r}\chi_e,
    \qquad
    w'_{j''}:=w_{j''}-2^{-r}\chi_e
    \qquad
    \text{for } e=\{j',j''\}\in M,
\]
and leave all other coordinates unchanged. Then
\[
    Hw'=Hw,
    \qquad
    \langle c,w'-w\rangle\le0,
\qquad
    \|A(w'-w)\|_\infty
    =
    2^{-r}\|Q\chi\|_\infty
    \le
    2^{-r}\beta.
\]
Moreover, all level-\(r\) bits are eliminated.

Repeating this for \(r=L,L-1,\ldots,1\), we obtain \(y\in\{0,1\}^n\) with
\[
    Hy=Hw=Hx,
    \qquad
    \langle c,y-w\rangle\le0,
\qquad
    \|A(w-y)\|_\infty
    \le
    \sum_{r=1}^{L}2^{-r}\beta
    \le \beta.
\]
Combining this with the bound on \(\|A(x-w)\|_\infty\), we obtain
\[
 Hy=Hx, \qquad
    \langle c,y-x\rangle\le0
, \qquad
    \|A(x-y)\|_\infty
    \le
    \beta+\frac{2}{n}\|A\|_\infty .
\]
The construction uses one discrepancy-oracle call per level, hence at most
\(L=\lceil 2\log_2 n\rceil\) calls.
\end{proof}

Using the Jiang--Reis algorithm~\cite{jiang2022tighter}
as the oracle gives the following bound in terms of the
hereditary discrepancy of \(A\).
The proof appears in Appendix~\ref{app:discHardConstraints}.

\begin{corollary} \label{cor:discHardConstraints} Let \(A\in\mathbb{R}^{m\times n}\), \(H\in\{0,1\}^{k\times n}\) with \(\|H_{\ast j}\|_1=1\) for all \(j\in[n]\), and \(c\in\mathbb{R}^n\). Let $\herdisc(A):= \max_{J \subseteq [n]} \disc(A_{J})$. Then, for every \(x\in[0,1]^n\) with \(Hx\in\mathbb Z^k\),
a randomized polynomial-time algorithm finds, with high probability, \(y\in\{0,1\}^n\) such that \[ Hy=Hx, \qquad \langle c,y-x\rangle\le0, \qquad \|A(x-y)\|_\infty \le O(1) \cdot \log (2+m) \cdot \herdisc(A) . \] \end{corollary}

\subsection{Mountain distance bound}

Many classical bounds in discrepancy theory control discrepancy
by a parameter of the columns. For example, the Beck--Fiala theorem gives
a bound in terms of the maximum column \(\ell_1\)-norm, while the
Komlós bound of~\cite{guo2026vector} gives
\(\disc(A)=O(\|A\|_2)\). In both cases, the relevant parameter measures the size
of each column, but it is insensitive to the order of the rows.
In this section we introduce a different column parameter, which captures
how much a column changes as we move from one row to the next. For a fixed
ordering of the rows, we define the mountain distance of \(A\) by
\[
    \mu(A)
    :=
    \max_{j\in[n]}
    \sum_{i\in[m-1]}
        |A_{i+1,j}-A_{i,j}|.
\]

One can think of a column \(A_{*j}\) as a mountain profile: the entry
\(A_{i,j}\) is the height at position \(i\). Then
$
    \sum_{i\in[m-1]} |A_{i+1,j}-A_{i,j}|
$
is the total vertical distance traveled while walking along this profile.
Hence \(\mu(A)\) is small when every column changes only a little along
the chosen row order. In particular, if \(\mu(A)=0\), then all rows of
\(A\) are identical.

For a \(0\)-\(1\) matrix,
\(\sum_{i\in[m-1]} |A_{i+1,j}-A_{i,j}|\) counts the number of switches in
column \(j\). For instance, a column of the form
\[
(0,\dots,0,
\underbrace{1,\dots,1}_{1},
0,\dots,0,\dots,0,
\underbrace{1,\dots,1}_{2},
0,\dots,0,\dots,0,
\underbrace{1,\dots,1}_{t},
0,\dots,0)
\]
has total variation \(2t\). Hence, if every column is a union of at most
\(t\) consecutive intervals, then \(\mu(A)\leq 2t\).

Unlike the Beck--Fiala parameter
\(\|A\|_1=\max_j\|A_{*j}\|_1\), the parameter \(\mu(A)\) sees the order of
the support. Indeed, one-block columns may be dense, with
\(\|A\|_1=\Theta(m)\), so the Beck--Fiala estimate gives only
\(O(\sqrt m)\), while their hereditary discrepancy is at most \(1\). Our
theorem gives
\[
    \disc(A)\leq O(t\log(2+\min\{m,n\})).
\]
In this sense, \(\mu(A)\) interpolates between consecutive-one columns and
unstructured sparse columns: 
\(\mu(A)\leq 2\|A\|_1\), but interval structure can make \(\mu(A)\) much
smaller than \(\|A\|_1\).

Since discrepancy is unchanged by permuting rows, it is natural to optimize
over the row order. We define 
\[
    \mu^{*}(A)
    :=
    \min_{\pi\in S_m}
    \max_{j\in[n]}
    \sum_{i\in[m-1]}
        |A_{\pi(i+1),j}-A_{\pi(i),j}|.
\]

The next theorem 
generalizes the beautiful rounding argument of
\cite{eisenbrand2013bin} for monotone matrices and shows that small
mountain distance implies small hereditary discrepancy, up to a logarithmic
factor. 
Matrices with small mountain distance arise naturally in rounding problems
where the constraints have an underlying order. For example, in scheduling
and resource-allocation problems, rows often correspond to time slots or
ordered resources, while each column corresponds to a job whose activity
occupies one interval, or a small number of intervals, in this order. Then
the column changes only at the endpoints of these intervals, and hence has
small total variation. Such interval-type structures appear in
\cite{bar2001unified,bar2005local,bar2005using}. In
Section~\ref{subsec:min-cost-rturn-chip-wiring}, we use this theorem together with Theorem~\ref{th:discHardConstraints} for
Min-Cost Bounded-Turn Unsplittable Flow in grids and tori.

For the theorem below, the row order is fixed. Applying the theorem after
an optimal row permutation gives the same bound with \(\mu^{*}(A)\) in
place of \(\mu(A)\).

\begin{theorem}\label{th:mountainDistUB}
Let \(A\in \mathbb{R}^{m\times n}\) and \(c\in\mathbb{R}^n\).
Then for every \(x \in [0,1]^{n}\) there exists
\(y \in \{0,1\}^{n}\) such that
$\langle c,y-x\rangle\leq 0$
and
\[
    \|A(x-y)\|_{\infty}
    \leq
    O(1)\cdot \mu(A)\log(2+\min\{m,n\})
    +
    \|A\|_{\infty}.
\]
Moreover, such a vector \(y\) can be found in polynomial time. Consequently,
\[
    \disc(A)
    \leq
    O(1)\cdot \mu(A)\log(2+\min\{m,n\})
    +
    2\|A\|_{\infty}.
\]
\end{theorem}

\begin{proof}
The algorithm is recursive. At every recursive call it receives a pair
\((A',x')\), where \(A'\) is obtained from \(A\) by deleting some rows,
while preserving their original order. Initially \(A'=A\) and \(x'=x\).
We maintain the invariants
\[
    A'x'=A'x
    \qquad\text{and}\qquad
    \langle c,x'-x\rangle\leq 0.
\]
Moreover, once a coordinate of \(x'\) becomes integral, it is never
changed again.

Consider one recursive call with input \((A',x')\). Let \(m'\) be the
number of rows of \(A'\). First replace \(x'\) by a vertex \(x''\) of the
polytope
\[
    P:=\{u\in[0,1]^n:
        A'u=A'x',
        \ u_j=x'_j \text{ for all } j \text{ with } x'_j\in\{0,1\}\}
\]
which minimizes \(\langle c,u\rangle\) over \(P\). Since \(x'\in P\), we
have $\langle c,x''\rangle\leq \langle c,x'\rangle.$
If \(n''\) denotes the number of
fractional coordinates of \(x''\), then $ n''\leq m'.$

If \(x''\in\{0,1\}^n\), we stop and output \(y=x''\). 
If \(m'=1\), then \(n''\leq1\). Let \(j\) be the remaining fractional
coordinate, if it exists. We round it in the direction which does not
increase the objective: if \(c_j\geq0\), set \(y_j=0\), and if \(c_j<0\),
set \(y_j=1\). Then 
$
    \langle c,y-x\rangle
    =
    \langle c,y-x''\rangle+\langle c,x''-x\rangle
    \leq0.
$
For the unique remaining row this creates error at most $\|A\|_{\infty}.$

Assume now that \(m'\geq2\). Let
$F:=\{j\in[n]:0<x''_j<1\}$
be the set of fractional coordinates, and define
$ Z(x''):=\{z\in\{0,1\}^n:
        z_j=x''_j \text{ for every } j\notin F\}.$
 Recall that $|F|=n''$, hence 
\[
\begin{aligned}
    \sum_{i\in[m'-1]}
    \max_{z\in Z(x'')}
    \left|
        \left\langle A'_{i+1,*}-A'_{i,*},z-x''\right\rangle
    \right|
    &\leq
    \sum_{i\in[m'-1]}\sum_{j\in F}
        |A'_{i+1,j}-A'_{i,j}|        \\
    &\leq
    n''\mu(A')                       \\
    &\leq
    \min\{m',n\}\mu(A).
\end{aligned}
\]
Here we used \(n''\leq m'\), \(n''\leq n\), and also that deleting rows
cannot increase \(\mu\). Indeed, if a row is deleted, then for every
column \(j\),
\[
    |A'_{i+1,j}-A'_{i-1,j}|
    \leq
    |A'_{i+1,j}-A'_{i,j}|
    +
    |A'_{i,j}-A'_{i-1,j}|,
\]
and therefore the total variation of each column cannot increase, thus $\mu(A'')\leq \mu(A')\leq \mu(A).$

By Markov's inequality, there exists a set
\(S\subseteq[m'-1]\) with
 $|S|\geq \frac{m'}2$
such that for every \(i\in S\),
\[
    \max_{z\in Z(x'')}
    \left|
        \left\langle A'_{i+1,*}-A'_{i,*},z-x''\right\rangle
    \right|
    \leq
    2\mu(A)\min\left\{1,\frac{n}{m'}\right\}.
\]
Choose \(T\subseteq S\) such that the pairs
\(\{i,i+1\}\), \(i\in T\), are pairwise disjoint and
 $|T|\geq \frac{|S|}{2}\geq \frac{m'}4.$
For every \(i\in T\), delete row \(i+1\) from \(A'\). Let the resulting
matrix be \(A''\). Then at most a \(3/4\)-fraction of the rows remains:
 $\#\operatorname{rows}(A'')\leq \frac34 m'.$
We recurse on \((A'',x'')\).

It remains to bound the error of the deleted rows. Suppose row \(i+1\)
was deleted in some recursive call and was paired with row \(i\), which
remained in the matrix. Let \(y\) be the final vector returned by the
algorithm. Since all coordinates integral in \(x''\) are fixed forever,
we have \(y\in Z(x'')\). Also, at the moment of deletion,
$A'x''=A'x.$
Therefore
\[
\begin{aligned}
    &
    \left|
        \langle A'_{i+1,*},y-x\rangle
        -
        \langle A'_{i,*},y-x\rangle
    \right|                                      \\
    &\qquad =
    \left|
        \left\langle A'_{i+1,*}-A'_{i,*},y-x''\right\rangle
    \right|                                      \\
    &\qquad \leq
    2\mu(A)\min\left\{1,\frac{n}{m'}\right\}.
\end{aligned}
\]
Thus every deleted row is connected to a row that survives that recursive
call, and their final errors differ by at most $2\mu(A)\min\left\{1,\frac{n}{m'}\right\},$
where \(m'\) is the number of rows at that call.

Now form a tree on the rows of \(A\) as follows. Whenever a row is deleted,
connect it to the row with which it was paired. If the algorithm stops
while several rows are still present, then all these remaining rows have
zero final error, so we connect them arbitrarily by zero-cost edges and
choose one of them as the root. If the algorithm stops with one row ($m'=1$ case) and
one fractional coordinate, the root has final error at most
$\|A\|_{\infty}$.
Thus, in all cases, the root error is at most \(\|A\|_{\infty}\).

Let \(m_t\) be the number of rows in the \(t\)-th recursive call. Since in
every call at most a \(3/4\)-fraction of the rows survives,
$m_{t+1}\leq \frac34 m_t.$
Along any path in the tree, the error accumulated at level \(t\) is at
most $2\mu(A)\min\left\{1,\frac{n}{m_t}\right\}.$
Therefore the total error accumulated along any root-to-leaf path is at
most $2\mu(A)\sum_t \min\left\{1,\frac{n}{m_t}\right\}.$
We claim that
\[
    \sum_t \min\left\{1,\frac{n}{m_t}\right\}
    =
    O\bigl(\log(2+\min\{m,n\})\bigr).
\]
Indeed, if \(m\leq n\), then every summand is at most \(1\), and there
are \(O(\log(2+m))\) levels. 
Now suppose \(m>n\). The levels with \(m_t>n\) contribute only \(O(1)\):
if \(m_t>n\), then the summand is \(n/m_t\), and since
\(m_{t+1}\leq \frac34 m_t\), these terms form a geometrically increasing
sequence bounded by a constant at the last such level. Once \(m_t\leq n\),
there are only \(O(\log(2+n))\) further levels, each contributing at most
\(1\). Hence the whole sum is \(O(\log(2+n))\), as required.

By the triangle inequality along the unique path from any row to the root,
for every original row \(r\), $|\langle A_{r,*},y-x\rangle|
    \leq
    \|A\|_{\infty}
    +
    O(1)\cdot \mu(A)\log(2+\min\{m,n\}).$
Together with \(\langle c,y-x\rangle\leq0\), this proves the rounding
statement.

Finally, the desired discrepancy bound follows as an easy consequence of the
rounding statement. Apply it with \(x=\frac12\mathbf 1\) and \(c=0\). We obtain
\(y\in\{0,1\}^{n}\) such that
\[
    \left\|
        A\left(\frac12\mathbf 1-y\right)
    \right\|_{\infty}
    \leq
    O(1)\cdot \mu(A)\log(2+\min\{m,n\})
    +
    \|A\|_{\infty}.
\]
Setting \(\chi:=\mathbf 1-2y\in\{\pm1\}^{n}\), we get $\|A\chi\|_{\infty}
    \leq
    O(1)\cdot \mu(A)\log(2+\min\{m,n\})
    +
    2\|A\|_{\infty}.$
\end{proof}

 Next we show that the bound above is best possible up to an absolute
constant factor.  The construction
is based on the lower bound of \cite{newman2012beck} for three permutations.

Given permutations
\(\pi_1,\dots,\pi_q\) of \([n]\), their matrix is the prefix-incidence
matrix \(P\in\{0,1\}^{qn\times n}\) whose rows are indexed by
\((a,t)\in[q]\times[n]\) and
\[
    P_{(a,t),j}
    :=
    \mathbf 1\bigl[j\in\{\pi_a(1),\dots,\pi_a(t)\}\bigr].
\]
Thus each block consists of the prefixes of one permutation.

\begin{theorem}[\cite{newman2012beck}]
Let \(\ell\in\mathbb{N}\) and \(n=3^\ell\). Then there exists a
matrix \(P\in\{0,1\}^{3n\times n}\) of three permutations such that $\disc(P)=\Omega(\log n).$
\end{theorem}

Next we show that a scaled matrix of three permutations gives tightness of
Theorem~\ref{th:mountainDistUB} up to a constant factor.

\begin{theorem}
Let \(\ell\in\mathbb{N}\), let \(n=3^\ell\), and let
\(\beta>0\). Then there exists a matrix
\(A\in \{0,\beta\}^{3n\times n}\) such that
$\mu(A)=\Theta(\beta)
    \text{ and }
    \disc(A)=\Theta(1)\cdot \mu(A) \cdot \log n .$
\end{theorem}

\begin{proof}
Let \(P\in\{0,1\}^{3n\times n}\) be the matrix of three permutations from
the theorem above, so $\disc(P)=\Omega(\log n).$
Set $A:=\beta P.$
Then
\[
     \disc(A)=\beta\disc(P)=\Omega(\beta\log n).
\]

We claim that \(\mu(P)=\Theta(1)\). Indeed, in each permutation block,
every column changes from \(0\) to \(1\) at most once, and there are only
two boundaries between the three blocks. Hence every column has total
variation \(O(1)\), so \(\mu(P)=O(1)\). Since \(n\geq 3\), some column
changes inside a permutation block, and therefore \(\mu(P)\geq1\). Thus $\mu(A)=\beta\mu(P)=\Theta(\beta).$
Consequently, $\disc(A)=\Omega(1)\cdot \mu(A) \cdot  \log n.$
The reverse inequality follows from Theorem~\ref{th:mountainDistUB}.
\end{proof}
\section{Min-Cost Unsplittable Flow}
\label{sec:min-cost-unsplittable-flow}

In this section we give rounding algorithms for min-cost unsplittable flow.
We consider two settings: arbitrary directed multigraphs with admissible paths
using at most $D$ hops, and multidimensional grids and tori with admissible
paths having at most $r$ turns. In both settings, starting from a feasible
fractional LP solution, our algorithms choose paths integrally, preserve the LP
objective value, and violate each capacity constraint only by an additive error.

\subsection{Min-Cost Unsplittable Flow on general directed multigraphs}\label{subsec:min-cost-unsplittable-flowGeneral}
Let $G=(V,E)$ be a directed multigraph on $n$ vertices with arc capacities
$C_e \in \mathbb{R}_{\ge 0}$ for every $e \in E$. We are given a multiset of
$k$ requests
\[
    \{(s_i,t_i,\rho_i)\}_{i \in [k]},
\]
where request $i$ requires routing an unsplittable flow of value
$\rho_i \in \mathbb{R}_{>0}$ from $s_i$ to $t_i$. For every request
$i \in [k]$, let $\mathcal{P}_i$ be the set of admissible $s_i$--$t_i$ paths,
each of length at most $D$. We regard paths as tagged by their request, so the
sets $\mathcal{P}_i$ are disjoint.
Let
$
    \mathcal{P} := \bigcup_{i \in [k]} \mathcal{P}_i .
$

For a path $P \in \mathcal{P}_i$, let $c_P \in \mathbb{R}$ denote the cost of
routing request $i$ on $P$. We do not assume that the costs are nonnegative;
thus the objective can encode both minimization and maximization problems.

The min-cost unsplittable-flow problem is to choose one path
$P_i\in\mathcal{P}_i$ for every request $i\in[k]$ so as to minimize
$\langle c,y\rangle$, subject to
\[
    L_e(y)
    :=
    \sum_{i=1}^k
    \sum_{\substack{P\in\mathcal{P}_i\\ e\in P}}
        \rho_i y_P
    \le C_e
    \qquad \forall e\in E,
\]
where $y_P\in\{0,1\}$ and
\[
    \sum_{P\in\mathcal{P}_i} y_P=1
    \qquad \forall i\in[k].
\]

\subsubsection{Related work}
We discuss related work in two settings: congestion-only rounding and
min-cost rounding. We call a result \emph{congestion-only} if it guarantees an
integral routing with bounded edge loads $L_e(y)$ but does not provide a
guarantee of the form $\langle c,y\rangle\le \langle c,x\rangle$ for a linear cost vector $c$. Let $\rho_{\max}:=\max_{i \in [k]} \rho_i$.
\paragraph{Congestion-only results.} 
Classical randomized rounding~\cite{raghavan1987randomized} and dependent
rounding~\cite{gandhi2006dependent} imply, for unit demands and
$C_e\gtrsim \log |E(G)|$,
$L_e(y)\le C_e+O(\sqrt{C_e\log |E(G)|})$. Leighton, Rao, and
Srinivasan~\cite{leighton1998multicommodity} used the Lovász Local Lemma as a
rounding tool for multicommodity routing. This approach was refined by Harris
and Srinivasan~\cite{harris2013moser}, who obtained
$L_e(y)\le C_e+O(\sqrt{C_e\log D})$ for $C_e\ge \log D$, and
$
    L_e(y)
    \le
    O\!\left(
        \frac{\log D}{1+\log(\log D/C_e)}
    \right)
$
for smaller $C_e$. Bansal and Nagarajan~\cite{bansal2017approximation} gave the
bound $L_e(y)
    \le
    C_e+
    \min\{
        O(\sqrt{D}\log n),
        O(\sqrt{C_e\log n}+\log^2 n)
    \}.$
For special graph classes, constant-additive congestion bounds of the form
\(L_e(y)\le C_e+O(\rho_{\max})\) are known. In particular, this type of
guarantee was obtained and refined for ring loading
\cite{schrijver1999ring,skutella2016note,daeubel2022improved} and for
outerplanar graphs \cite{aleman2025unsplittable,aleman2026pinning}.

\paragraph{Min-cost results.}
Fewer results also preserve a cost objective. Our result applies
to arbitrary directed multigraphs with admissible paths using at most $D$ hops
and arbitrary path costs on the LP support. It gives
$$
    L_e(y)
    \le
    C_e+
    O\!\left(\rho_{\max}\sqrt D\right)
    ,
    \qquad
    \langle c,y\rangle\le \langle c,x\rangle .
$$

For single-source unsplittable min-cost flow in general directed graphs,
Skutella~\cite{skutella2002approximating} gives an objective-preserving
rounding with $L_e(y)\le 2C_e+\rho_{\max}.$
If the demands are ordered by divisibility, this improves to $L_e(y)\le C_e+\rho_{\max}.$
For planar single-source instances with arbitrary demands, Traub, Vargas Koch,
and Zenklusen~\cite{traub2024single} prove the objective-preserving bound $L_e(y) \le C_e+2 \rho_{\max}.$
For further objective-preserving results in structured settings, see
\cite{majthoub2025integer}.
\subsubsection{Rounding algorithm}
The min-cost unsplittable-flow problem is the following integer program:
\begin{align}
    \min \quad
        & \langle c,x\rangle
        \label{ip:objective} \\
    \text{s.t.} \quad
        & \sum_{P \in \mathcal{P}_i} x_P = 1
        && \forall i \in [k],
        \label{ip:assignment} \\
        & \sum_{i=1}^k
          \sum_{\substack{P \in \mathcal{P}_i \\ e \in P}}
            \rho_i x_P \le C_e
        && \forall e \in E,
        \label{ip:capacity} \\
        & x_P \in \{0,1\}
        && \forall P \in \mathcal{P}.
        \label{ip:integrality}
\end{align}

The standard LP relaxation is obtained by replacing
\eqref{ip:integrality} with nonnegativity constraints:
\begin{align}
        & x_P \ge 0
        && \forall P \in \mathcal{P}.
        \label{lp:nonnegative}
\end{align}

For a feasible LP solution $x$, let
\[
    \mathcal{P}'_i:=\{P\in\mathcal{P}_i:x_P>0\},
    \qquad
    \mathcal{P}':=\bigcup_{i\in[k]}\mathcal{P}'_i .
\]
For $z\in\mathbb{R}^{\mathcal{P}'}$, write
\[
    L_e(z):=
    \sum_{i=1}^k
    \sum_{\substack{P\in\mathcal{P}'_i\\ e\in P}}
        \rho_i z_P .
\]
Also define 
\[
    \rho_{\max}^e
    :=
    \max\{\rho_i : i\in[k],\; \exists P\in\mathcal{P}'_i
    \text{ with } e\in P\},
\]
with $\rho_{\max}^e:=0$ if the set is empty.

We use the following standard facts about solving the LP relaxation. If the path
family \(\mathcal P\) is given explicitly, then the LP has \(|\mathcal P|\)
variables and \(k+|E(G)|\) constraints, and can be solved in time polynomial in
\(|\mathcal P|\), \(|E(G)|\), and \(k\). In this explicit case, arbitrary path
costs \(c_P\in\mathbb R\) are allowed.

For implicit path families, we use the standard pricing-oracle approach; see,
e.g., \cite{chakrabarti2007approximation,fan2024multi}. This gives polynomial
time solvability whenever the relevant minimum-cost path problem over
\(\mathcal P_i\) can be solved efficiently. In particular, for edge-additive
nonnegative costs, this applies to the family of all \(s_i\)--\(t_i\) paths and
to the family of all such paths with at most \(D\) hops, the latter via the
standard layered-graph construction. If \(G\) is a DAG, the same conclusion
holds for arbitrary real edge-additive costs.

Finally, request-independent per-unit edge costs are captured by setting
\[
    \mathrm{cost}(e,i):=\rho_i\mathrm{cost}(e),
    \qquad
    c_P:=\rho_i\sum_{e\in P}\mathrm{cost}(e).
\]

The proof of the next rounding theorem applies
Theorem~\ref{th:discHardConstraints}.

\begin{theorem}[General directed multigraphs with $D$-hop paths]
\label{thm:min-cost-unsplittable-rounding}
Assume that $G$ is an arbitrary directed multigraph and every admissible path
$P\in\mathcal{P}_i$ uses at most $D$ hops. Let $x$ be a feasible solution to
the LP relaxation \eqref{ip:objective}--\eqref{lp:nonnegative}. Then one can
find, in polynomial time, a vector $y\in\{0,1\}^{\mathcal{P}'}$ such that
\begin{align}
    &\sum_{P\in\mathcal{P}'_i} y_P = 1
    && \forall i\in[k],
    \label{eq:rounding-assignment} \\
    &L_e(y)
    \le
    L_e(x)
    +
    O\!\left(\rho_{\max}^e\sqrt D\right)
    && \forall e\in E,
    \label{eq:rounding-capacity} \\
    &\langle c,y-x\rangle \le 0 .
    \label{eq:rounding-cost}
\end{align}
\end{theorem}

The required discrepancy oracle is provided by the following
result of Guo, Fang, and Lu.

\begin{theorem}[Guo--Fang--Lu~\cite{guo2026vector}]
\label{th:KomlosByGuo}
Let $M \in \mathbb{R}^{m \times n}$ satisfy $\|M\|_2 \le 1$.
Then $\disc(M) \le O(1),$
and the corresponding coloring can be found in polynomial time.
\end{theorem}

\begin{proof}[Proof of Theorem~\ref{thm:min-cost-unsplittable-rounding}]
Let $p:=|\mathcal{P}'|$ and index the coordinates by paths
$P\in\mathcal{P}'$. Define
$A \in \mathbb{R}_{\ge 0}^{E \times \mathcal{P}'}$ by $ A_{eP}:=\rho_i\mathbf{1}[e\in P]
    \text{ for } P\in\mathcal{P}'_i .$
Thus $L_e(z)=\langle A_{e\ast},z\rangle$ for every
$z\in\mathbb{R}^{\mathcal{P}'}$.

For every $e\in E(G)$, set
\[
    B_{e\ast}:=
    \begin{cases}
        A_{e\ast}/\rho_{\max}^e, & \rho_{\max}^e>0,\\
        0, & \rho_{\max}^e=0 .
    \end{cases}
\]
Then $\|B\|_\infty\le 1$.

Let $H\in\{0,1\}^{[k]\times\mathcal{P}'}$ be the request-path incidence matrix, $H_{iP}:=\mathbf{1}[P\in\mathcal{P}'_i].$
Since $x$ is feasible, $Hx=\mathbf{1}$.
If \(D=0\), choosing a cheapest supported path for each request
proves the claim, so assume \(D\ge1\).
Fix \(Q\in\mathcal Q(B)\). Each column of \(Q\) has at most
\(2D\) nonzero entries, all in \([-1,1]\), so \(\|Q\|_2\le\sqrt{2D}\).
Let \(\tau\) be the least power of two with \(\tau^2\ge2D\).
Then \(\|Q/\tau\|_2\le1\) and \(\tau=O(\sqrt D)\).
Theorem~\ref{th:KomlosByGuo}, applied to \(Q/\tau\), therefore
provides a discrepancy oracle with
$$\beta=C\tau=O(\sqrt D).$$

Applying Theorem~\ref{th:discHardConstraints} to \(B,H,c,x\),
we obtain \(y\in\{0,1\}^{\mathcal P'}\) such that
\[
    Hy=\mathbf1,
    \qquad
    \langle c,y-x\rangle\le0,
    \qquad
    \|B(x-y)\|_\infty
    \le \beta+\frac2p\|B\|_\infty
    =O(\sqrt D),
\]
where we used \(\|B\|_\infty\le1\).
Multiplying the bound for row \(e\) by \(\rho_{\max}^e\) gives
\[
    L_e(y)\le L_e(x)+O\!\left(\rho_{\max}^e\sqrt D\right).
\]
Finally, \(Hy=\mathbf1\) selects exactly one path per request.
\end{proof}

\subsection{Min-Cost Bounded-Turn Unsplittable Flow in Grids and Tori}\label{subsec:min-cost-rturn-chip-wiring}

In this subsection, $G$ is the $d$-dimensional undirected grid on vertex set
$[n]^d$, or the corresponding $d$-dimensional torus. For every request
$i \in [k]$, the admissible family $\mathcal{P}_i$ consists of
$s_i$--$t_i$ paths with at most $r$ turns, i.e. each path is a concatenation of
at most $r+1$ straight coordinate segments. 

\subsubsection{Related work}
In the VLSI-routing literature, one typically assumes unit demands
$\rho_i=1$ for all $i\in[k]$ and disjoint source-sink pairs, and studies the
maximum edge congestion $\max_e L_e(y)$. Grid routing is a classical abstraction
of VLSI global routing~\cite{hu2001survey,kahng2011vlsi}.

Karp, Leighton, Rivest, Thompson, Vazirani, and Vazirani
\cite{karp1987global}  proved that there are instances on the $n\times n$ grid for which every
one-turn routing has congestion at least $\left\lfloor n/2 \right\rfloor+1,$
and gave a polynomial-time algorithm producing a one-turn routing with $L_e(y)\le \left\lfloor n/2 \right\rfloor+2$.

Bhatia, Leighton, Makedon, and Norton~\cite{bhatia1992improved}  gave,
for terminal pairs at grid distance at most \(D\), a fast two-bend routing
algorithm with
$L_e(y)\le D+1$, which is within additive $1$ of a natural
worst-case lower bound. For one-bend routing, they  obtained a polynomial-time routing with $L_e(y)\le C_e+O(\log^2 D).$

Hajiaghayi, Kleinberg, and Leighton~\cite{hajiaghayi2006semi}  constructed an
unsplittable-flow instance on the $n\times n$ grid with fractional congestion
$1$  such that any routing using fewer than $r$ turns has congestion $\Omega\!\left(n^{1/(2r+1)}\right).$
We refer to \cite{zuzic2023sparse} for a review of results on online versions of the problem.

\subsubsection{Rounding algorithm}
Our result gives a cost-preserving LP-rounding guarantee for bounded-turn
routing in $d$-dimensional grids and tori for an arbitrary multiset of requests $\{(s_i,t_i,\rho_i)\}_{i \in [k]}$.
\begin{theorem}[$r$-turn paths in grids and tori]
\label{thm:min-cost-rturn-grid-torus-rounding}
Assume that $G$ is the $d$-dimensional undirected grid or torus on vertex set
$[n]^d$, and every admissible path $P\in\mathcal{P}_i$ has at most $r$ turns.
Let $x$ be a feasible solution to the LP relaxation
\eqref{ip:objective}--\eqref{lp:nonnegative}. Then one can find, in polynomial
time, a vector $y\in\{0,1\}^{\mathcal{P}'}$ such that
\begin{align}
    &\sum_{P\in\mathcal{P}'_i} y_P = 1
    && \forall i\in[k],
    \label{eq:rturn-rounding-assignment} \\
    &L_e(y)
    \le
    L_e(x)
    +
    O\!\left(
        \rho_{\max}^e (r+1)\log |V(G)|
    \right)
    && \forall e\in E,
    \label{eq:rturn-rounding-capacity} \\
    &\langle c,y-x\rangle \le 0 .
    \label{eq:rturn-rounding-cost}
\end{align}
\end{theorem}

\begin{proof}
We prove the statement for grids. 
Order the edges by coordinate direction, then by coordinate line, and then by
position on that line.

First assume that $\rho_i=1$ for all $i\in[k]$. Let
$p:=|\mathcal{P}'|$, and define the edge-path matrix
$A\in\{0,1\}^{E\times\mathcal{P}'}$ by
$
    A_{eP}:=\mathbf{1}[e\in P].
$
Let $H\in\{0,1\}^{[k]\times\mathcal{P}'}$ be the request-path incidence matrix,
$
    H_{iP}:=\mathbf{1}[P\in\mathcal{P}'_i].
$
Since $x$ is feasible, $Hx=\mathbf{1}$.

We bound the discrepancy term in Theorem~\ref{th:discHardConstraints}. Fix
$Q\in\mathcal{Q}(A)$. Each column of $Q$ is of the form
$A_{\ast P}-A_{\ast P'}$. Since every admissible path has at most $r$ turns, it
is the union of at most $r+1$ straight coordinate segments, and each such
segment is an interval in the chosen edge ordering. Hence each path column has
mountain distance at most $2(r+1)$, and therefore
\[
    \mu(Q)\le 4(r+1),
    \qquad
    \|Q\|_\infty\le 1 .
\]
Moreover, by the definition of $\mathcal{Q}(A)$, every such $Q$ has at most
$|E(G)|$ columns. Applying Theorem~\ref{th:mountainDistUB} to all column
submatrices of $Q$ gives
$$
    \disc(Q)\le O((r+1)\log |E(G)|).
$$

Applying Theorem~\ref{th:discHardConstraints} to $A,H,c,x$, we obtain
$y\in\{0,1\}^{\mathcal{P}'}$ such that
\[
    Hy=\mathbf{1},
    \qquad
    \langle c,y-x\rangle\le 0, \qquad \|A(y-x)\|_\infty
    \le O((r+1)\log |E(G)|).
\]
This proves the theorem when all demands are equal to $1$.

We now handle arbitrary demands by dyadic bucketing. For every integer $\ell$,
let $I_\ell:=\{i\in[k]:2^{\ell-1}<\rho_i\le 2^\ell\}.$
For each nonempty class $I_\ell$, let
\[
    \mathcal{P}'_\ell:=\bigcup_{i\in I_\ell}\mathcal{P}'_i,
    \qquad
    A^\ell_{eP}:=\rho_i\mathbf{1}[e\in P]
    \quad
    \text{for } P\in\mathcal{P}'_i,\; i\in I_\ell .
\]
Apply the unit-demand argument to the scaled matrix
$\overline{A}^\ell:=A^\ell/ 2^{\ell}$. Since
$\rho_i/ 2^{\ell}\in(1/2,1]$ for $i\in I_\ell$, the same mountain-distance
bound gives
\[
    \max_{Q\in\mathcal{Q}(\overline{A}^\ell)}\disc(Q)
    \le O((r+1)\log |E(G)|).
\]
Thus Theorem~\ref{th:discHardConstraints}, applied to
$\overline{A}^\ell,H^\ell,c^\ell,x^\ell$, gives an integral vector
$y^\ell$ selecting one path for each request in $I_\ell$, with
\[
    \langle c^\ell,y^\ell-x^\ell\rangle\le 0, \qquad \|A^\ell(y^\ell-x^\ell)\|_\infty
    \le
    O((r+1) 2^{\ell}\log |E(G)|).
\]

Let $y$ be the concatenation of the vectors $y^\ell$. Then one path is selected
for every request, and summing the objective guarantees yields $\langle c,y-x\rangle\le 0.$
For a fixed edge $e$, only classes containing a request with a supported path
using $e$ can affect its load. 

Therefore
\[
    L_e(y)-L_e(x)
    \le
    \sum_{\ell:\, I_\ell \text{ contributes to } e}
        O((r+1) 2^{\ell}\log |E(G)|)
    \le
    O((r+1)\rho_{\max}^e\log |E(G)|).
\]
Since $\log |E(G)|=O(\log (dn^d))=O(\log |V(G)|)$, the claim follows. 

The proof for tori is the same: after fixing
one cut on each coordinate cycle, every straight torus segment is the union of
at most two intervals in the induced ordering, changing only the absolute
constant.
\end{proof}

\section{Rainbow prefix  hypergraph decomposition}\label{sec:RainbowPrefix}

The rainbow prefix hypergraph decomposition problem is defined as follows.
Let \(\mathcal F\) be an \(N\)-regular multihypergraph whose edges are colored
with \(n\) colors, $\mathcal F=\mathcal C_1\sqcup\cdots\sqcup\mathcal C_n,$
where each color class contains exactly \(N\) edges. The goal is to find a
chain $ \varnothing=\mathcal F_0
    \subset \mathcal F_1
    \subset \cdots
    \subset \mathcal F_N=\mathcal F$
such that
\[
    d_{\mathcal F_i}(v)\in[i-\beta,i+\beta]
    \qquad \forall i\in[N],\ \forall v,
\]
and each difference of consecutive prefixes is rainbow:
\[
    |\mathcal C_j\cap(\mathcal F_i\setminus \mathcal F_{i-1})|=1
    \qquad \forall i \in [N],j\in[n].
\]

Note that, when $\beta=0$, the rainbow prefix hypergraph decomposition problem is precisely the problem of decomposing $\mathcal F$ into rainbow perfect matchings.

\subsection{Related work}

The decomposition of graphs and hypergraphs into prescribed substructures is a
central theme in extremal and probabilistic combinatorics; for modern surveys,
see Keevash~\cite{keevash2018hypergraph} and Glock, K{\"u}hn and
Osthus~\cite{glock2021extremal}. A foundational exact decomposition result is
Baranyai's theorem
\cite{baranyai1974factrization,baranyai1979edge}, which states that if
\(k\mid m\), then the complete \(k\)-uniform hypergraph
\(\binom{[m]}{k}\) decomposes into perfect matchings. Approximate decomposition
and near-perfect matching results go back to the nibble method of
R{\"o}dl~\cite{rodl1985packing} and to the theorem of Pippenger and
Spencer~\cite{pippenger1989asymptotic}. More recently, Keevash proved the
existence of designs~\cite{keevash2016existence}, and Glock, K{\"u}hn, Lo and
Osthus gave an iterative-absorption proof of the existence of hypergraph
\(F\)-designs for arbitrary \(F\)~\cite{glock2023existence}.

Rainbow one-factors in complete hypergraphs were
studied by El-Zanati, Plantholt, Sissokho and Spence
\cite{el2007existence}. They proved that if \(\binom{[mk]}{k}\) is given
together with a one-factorization, then there exists a one-factor whose edges
come from distinct factors. Li and Xu~\cite{li2007existence} generalized this
by showing that a rainbow perfect matching exists in any proper coloring of
\(\binom{[mk]}{k}\).
 Aharoni and
Berger~\cite{aharoni2009rainbow} studied rainbow matchings in \(r\)-partite
\(r\)-uniform hypergraphs, where one seeks a matching using at most one edge
from each color class. Keevash and Yepremyan~\cite{keevash2018rainbow}
studied rainbow matchings in properly colored multigraphs. Gao, Ramadurai,
Wanless and Wormald~\cite{gao2021full} studied full rainbow matchings in graphs. Coulson, Keevash, Perarnau and Yepremyan~\cite{coulson2020rainbow}
proved rainbow clique-factors results for graphs with suitably bounded colorings.

Finally, the prefix condition connects the problem to prefix discrepancy.
Classical discrepancy controls only the final imbalance of a coloring, while
prefix discrepancy asks for all initial segments to remain balanced. Bansal,
Jiang, Meka, Singla and Sinha \cite{bansal2022prefix} studied
prefix discrepancy and related vector-balancing problems. Bansal, Rohwedder and
Svensson \cite{bansal2022flow} introduced a prefix Beck--Fiala
framework motivated by scheduling. In the present setting, the prefix constraint is
$\max_{v,i}|d_{\mathcal F_i}(v)-i|\leq \beta,$
but it must be achieved together with the exact rainbow constraints
$|\mathcal C_j\cap(\mathcal F_i\setminus\mathcal F_{i-1})|=1.$
Thus the problem combines two features that are usually treated separately:
rainbow transversal constraints and prefix discrepancy constraints.

\subsection{Rainbow prefix decomposition theorem}

The proof of Theorem~\ref{th:prefixDecomp} uses
Theorem~\ref{th:KomlosByGuo} as the discrepancy oracle.
Using Beck--Fiala~\cite{beck1981integer,lovasz1986discrepancy}
instead gives \(\disc(Q)\le2\|Q\|_1\le4k\) at each split.
The same proof then yields \(\beta\le(8k+4)\log N\),
giving a small explicit constant for the subsequent
\(2\)-multicast application.

\begin{theorem}\label{th:prefixDecomp}
Let $N$ be a power of two.
Let $\mathcal{C}_{1}, \ldots ,\mathcal{C}_{n} \subseteq \binom{[kn]}{ k}$ be multihypergraphs each consisting exactly of $N$ edges.
Let $\mathcal{F} = \bigsqcup_{i \in [n]} \mathcal{C}_{i}$ be an $N$-regular multihypergraph with $n \cdot N$ edges. Then, in time polynomial in \(n, N\) and \(k\), we can find prefix decomposition $\varnothing =\mathcal{F}_0 \subset \mathcal{F}_1 \subset  \dots \subset \mathcal{F}_N=\mathcal{F}$ such that 
\begin{enumerate}
    \item $d_{\mathcal{F}_i}(v) \in [i -\beta, i + \beta]$ for all $v \in [kn]$ and $i \in [N]$ with $\beta \leq O(\sqrt{k}\log N)$,
    \item $|\mathcal{C}_j \cap \mathcal{F}_{i}| = i$ for all $i\in [N], j \in [n]$.
\end{enumerate}

\end{theorem}
\begin{proof}
We treat repeated edges as distinct labelled copies. 
We first record the splitting step. Let \(\mathcal G\subseteq \mathcal F\) be such
that, for some even \(s\), $|\mathcal C_j\cap \mathcal G|=s
    $ for all $j\in[n]$.
Let \(E(\mathcal G)\) be the set of edge-copies of \(\mathcal G\). Let $A\in\{0,1\}^{[kn]\times E(\mathcal G)}$
be the vertex-edge incidence matrix, and let $H\in\{0,1\}^{[n]\times E(\mathcal G)}$
be the color-incidence matrix: $H_{j,e}=1  \Longleftrightarrow  e\in\mathcal C_j .$
Each column of \(H\) has exactly one \(1\). Take \(x=\frac12\mathbf 1\). Then $Hx=\frac{s}{2}\mathbf 1\in\mathbb Z^n.$
Every column of \(Q\in\mathcal Q(A)\) has at most \(2k\)
nonzero entries, all in \(\{-1,1\}\), so
\(\|Q\|_2\le\sqrt{2k}\).
Let \(\tau\) be the least power of two with \(\tau^2\ge2k\).
Applying Theorem~\ref{th:KomlosByGuo} to \(Q/\tau\) provides
a discrepancy oracle with bound \(C\tau\), where \(C\) is
the absolute constant in that theorem.
Theorem~\ref{th:discHardConstraints}, applied with \(c=0\), gives
\[
    Hy=Hx,
    \qquad
    \|A(x-y)\|_\infty
    \le C\tau+\frac{2}{|E(\mathcal G)|}
    \le C\tau+2=:\rho.
\]
Since \(\tau<2\sqrt{2k}\), we have \(\rho=O(\sqrt{k})\).
Define $\mathcal G^0=\{e:y_e=1\},
    \mathcal G^1=\{e:y_e=0\}.$
Then $\mathcal G=\mathcal G^0\sqcup \mathcal G^1,$
\[
    |\mathcal C_j\cap\mathcal G^0|
    =
    |\mathcal C_j\cap\mathcal G^1|
    =
    \frac{s}{2}
    \qquad \text{for all }j\in[n],
\]
and, for every vertex \(v\),
\[
    \left|
        d_{\mathcal G^b}(v)-\frac12 d_{\mathcal G}(v)
    \right|
    \le \rho,
    \qquad b\in\{0,1\}.
\]

Now build a complete binary tree of depth $\log N$. Put $\mathcal F^\emptyset=\mathcal F.$
If \(\alpha\in\{0,1\}^t\), \(t< \log N\), has already been constructed, then $|\mathcal C_j\cap\mathcal F^\alpha|=\frac{N}{2^t}
     \text{ for all }j\in[n],$
and this number is even. Apply the splitting step and write $\mathcal F^\alpha
    =
    \mathcal F^{\alpha0}\sqcup \mathcal F^{\alpha1}.$
Thus every leaf contains exactly one edge from each \(\mathcal C_j\).

For a node \(\alpha\), write $D_\alpha:=\frac{N}{2^{|\alpha|}}.$
We claim that
\[
    \left|d_{\mathcal F^\alpha}(v)-D_\alpha\right|\le 2\rho
    \qquad
    \text{for every }v\in[kn].
\]
At the root this is exact, since \(\mathcal F\) is \(N\)-regular. If the claim holds
for \(\alpha\), then for \(b\in\{0,1\}\),
\[
\begin{aligned}
    \left|d_{\mathcal F^{\alpha b}}(v)-D_{\alpha b}\right|
    &\le
    \left|
        d_{\mathcal F^{\alpha b}}(v)
        -\frac12 d_{\mathcal F^\alpha}(v)
    \right|
    +
    \frac12
    \left|d_{\mathcal F^\alpha}(v)-D_\alpha\right|  \\
    &\le \rho+\frac12\cdot 2\rho
    =2\rho .
\end{aligned}
\]

Order the leaves lexicographically: $\mathcal L_1,\ldots,\mathcal L_N.$
Define $ \mathcal F_i:=\mathcal L_1\sqcup\cdots\sqcup \mathcal L_i$ for $
     i=1,\ldots,N.$
Since every leaf contains exactly one edge from every \(\mathcal C_j\), $|\mathcal C_j\cap(\mathcal F_i\setminus \mathcal F_{i-1})|=1
    \text{ for all }j\in[n],\ i\in[N].$
Equivalently, $|\mathcal C_j\cap\mathcal F_i|=i
    \text{ for all }j\in[n],\ i\in[N].$

It remains to bound the degrees of prefixes. The first \(i\) leaves are a disjoint
union of at most \(\log N\) complete subtrees. Let \(\mathcal P_i\) be the
corresponding set of nodes. Then $\mathcal F_i=\bigsqcup_{\alpha\in\mathcal P_i}\mathcal F^\alpha$ and $\sum_{\alpha\in\mathcal P_i}D_\alpha=i$.
Therefore, for every \(v\in[kn]\),
\[
\begin{aligned}
    |d_{\mathcal F_i}(v)-i|
    &=
    \left|
        \sum_{\alpha\in\mathcal P_i}
        \bigl(d_{\mathcal F^\alpha}(v)-D_\alpha\bigr)
    \right| \\
    &\le
    \sum_{\alpha\in\mathcal P_i}
    \left|d_{\mathcal F^\alpha}(v)-D_\alpha\right|  \\
    &\le
    2\rho\log_2 N.
\end{aligned}
\]
Thus we may take
\(\beta=2\rho\log N=O(\sqrt{k}\log N)\).
\end{proof}
\section{2-multicast scheduling via rainbow prefix decomposition}\label{sec:twoMulticast}

The following optimization problem came about from an application in Telecom. The goal is to minimize worst-case buffer use in the following scheduling problem. 
We are given a chip with $n$ in-ports and $2n$ out-ports. At each time moment (round), the  in-ports receive a message (all assumed to be of the same size), each of which is assigned
two out-ports as destinations. We are given the composite communication request (CCR) for
$n$ rounds. 
CCR is an ensemble of $n$ messages for each of $n$ in-ports. Each of
the messages has two destinations. We add a natural condition that each of the $2n$ out-ports
is the destination for exactly $n$ messages.

Our goal is to split the CCR into rounds, each being as described above. If a received
message cannot be sent to all of its destinations immediately, the undelivered copies are
stored in the chip's buffer and sent later in some  queueing order. We seek an algorithm that, for every CCR satisfying the above
balance condition, constructs a feasible schedule while keeping the maximum buffer
occupancy as small as possible.

\subsection{Scheduling algorithm}

Without loss of generality, denote the set of inports \([n]\).
The CCR has size
\[
    (\#\text{ in-ports},\#\text{ out-ports},\#\text{ rounds})
    =
    n\times 2n\times n .
\]

A convenient way to represent  a CCR is as follows.
Consider a multigraph $G$ on $2n$ vertices, one vertex per  out-port. Each
message is an edge in this graph, connecting its two destinations. Each edge has a color, corresponding to the in-port that sends it. There
are $n$ edges of each of the $n$ colors. Moreover, the degree of each
vertex is $n$, and the graph may have parallel edges.  

Consider a CCR and a communication scheme for it (that is, a partition of the incoming
messages into rounds). We describe the communication scheme by a sequence of partial graphs
$G^1,\ldots,G^n$, where $G^i$ consists of all edges that arrive in the first $i$ rounds. 

$G^\ell$ is a graph with $\ell$ edges of each of the $n$
colors, and  average degree $\ell$, and $G^n=G$.

The next goal is
to get a description of the buffer size required for a CCR with a concrete communication scheme. To
this end, let us consider one out-port $v$ and understand, how much buffer space do we need for
messages with $v$ as the destination. We use the notation  $d_v^\ell$ for the degree of $v$ in $G^\ell$ $q_v^\ell$ for {\it excess degree}: 
$q_v^\ell:=d_v^\ell-\ell$.

\begin{lemma}
Let $M_b(v)=\max_{a:0\leq a\leq b}(q_v^b-q_v^a)$.
\begin{enumerate}
    \item The number of messages destined to $v$ that will be stored in the buffer at round $b$
    is $M_b(v)$.
    \item The total number of messages in the buffer at round $b$ is
    $-\sum_v\min_{a:0\leq a\leq b}q_v^a$.
    \item The maximum number of messages in the buffer is
    \[
        -\sum_v\min_{a:0\leq a\leq n}q_v^a.
    \]
\end{enumerate}
\end{lemma}

\begin{proof}
\begin{enumerate}
    \item Let us show that the size of the buffer is at least $M_b$. Clearly, the number of
    messages received for $v$ between moments $a$ and $b$ is $d_v^b-d_v^a$. At the same time,
    we have $b-a$ moments to send them out to $v$, which shows that at moment $b$ the buffer
    will contain at least $q_v^b-q_v^a$ messages for $v$. Next, assume that $M'$ is the
    number of messages stored for $v$ in the buffer at moment $b$. Let $a$ be the last moment
    before $b$ at which the buffer had no messages for $v$. It is easy to see that then
    $q_v^b-q_v^a=M'$.

    \item Using the first part, the number of messages in the buffer at round $b$ is
    \[
    \sum_v M_b(v)
      = \sum_v q_v^b + \sum_v \max_{a:0\leq a\leq b}(-q_v^a)
      = \sum_v \max_{a:0\leq a\leq b}(-q_v^a)
      = -\sum_v \min_{a:0\leq a\leq b}q_v^a.
    \]
    In the second equality we used the fact that the sum of excess degrees is $0$ at any stage.

    \item It is easy to see that the last displayed expression is maximized when $b=n$.
\end{enumerate}
\end{proof}

Now, we can reformulate this scheduling problem  as the following prefix
decomposition problem. \begin{quote} Let $G=\mathcal C_1\sqcup\cdots\sqcup \mathcal C_n$
be an edge-colored multigraph on vertex set \([2n]\), where the vertices are the
out-ports and the color classes correspond to the active in-ports. Each color
class \(\mathcal C_j\) consists of exactly \(n\) edges, and \(G\) is \(n\)-regular: $d_G(v)=n
    \text{ for every }v\in[2n].$
We want to find a chain
$\varnothing=G^0\subset G^1\subset\cdots\subset G^n=G$
such that $|\mathcal C_j\cap G^\ell|=\ell
    \text{ for every }j,\ell\in[n].$\end{quote}
The quantity $q_v^\ell:=d_{G^\ell}(v)-\ell$
measures the prefix excess at output \(v\). For a given chain \(G^\bullet\), the
required total buffer size is $B(G^\bullet)
    =
    -\sum_{v\in[2n]}\min_{0\leq \ell\leq n} q_v^\ell .$
Hence the original scheduling objective is
\[
    \min_{G^\bullet}
    B(G^\bullet)
    =
    \min_{G^\bullet}
    \left(
        -\sum_{v\in[2n]}\min_{0\leq \ell\leq n}
        \bigl(d_{G^\ell}(v)-\ell\bigr)
    \right),
\]
where the minimum is over all chains satisfying the color constraints.

In fact, we can prove a stronger
prefix-discrepancy guarantee $|q_v^\ell|
    =
    |d_{G^\ell}(v)-\ell|
    \leq \beta
    \text{ for every }v\in[2n]\text{ and }\ell\in[n].$
Since there are \(2n\) out-ports, the buffer identity gives $B(G^\bullet)
    \leq
    2n\beta .$
Thus any such prefix decomposition gives a schedule with maximum
buffer occupancy at most \(2n\beta\).
Theorem~\ref{th:prefixDecomp}, applied with \(k=2\), gives
the bound \(O(n\log n)\).
Using the Beck--Fiala oracle instead gives
\(\beta\le20\log n\), yielding the following explicit bound.

\begin{theorem}
Let \(n\) be a power of two. For every balanced \(n\times 2n\times n\) CCR,
there is a polynomial-time scheduling algorithm whose maximum buffer occupancy is
at most $40n\log n .$
\end{theorem}

\section{Conclusions and future directions}\label{sec:Conclusions}
Several natural questions remain open.

\begin{itemize}
    \item \textbf{Two-block \(0\)-\(1\) matrices.}
    Our lower bound for the mountain-distance theorem comes from the
three-permutations construction, where the support of each column may split into
three intervals of consecutive ones in the row order. We conjecture that this is
necessary for logarithmic discrepancy: if \(A\) is a \(0\)-\(1\) matrix such
that, after some ordering of the rows, the ones in every column form at most two
consecutive blocks, then
\[
    \disc(A)=O(1).
\]
\item \textbf{Weighted one-block matrices.}
Let \(A\in\{0,1\}^{m\times n}\), and define
\[
    \herdisc(A):=\max_{J\subseteq[n]}\disc(A_{\ast J}),
    \qquad
    \wdisc(A):=\sup_{d\in[0,1]^n}\disc(A\operatorname{diag}(d)).
\]
Reis~\cite{reis2023vector} conjectures that
\[
    \wdisc(A)\le O(1)\cdot \herdisc(A)
\]
for every \(0\)-\(1\) matrix \(A\). Our weighted one-block question is a special
case: if the ones in every column of \(A\) form one consecutive block, then
\(\herdisc(A)\le 1\). Thus the conjecture would imply that every weighted
one-block matrix
\[
    A_{ij}=d_j\mathbf 1[i\in I_j],
    \qquad d_j\in[0,1],
\]
has constant hereditary discrepancy. Such matrices arise naturally in
resource-allocation rounding problems. We ask whether this special case admits
a direct proof.
    \item \textbf{Better rainbow prefix decompositions.}
Using Theorem~\ref{th:KomlosByGuo}, our framework gives
prefix error \(O(\sqrt{k}\log N)\) constructively.
Banaszczyk's bound on the Steinitz constant~\cite{banaszczyk2012series}
also yields the non-constructive bound \(O(k\sqrt{\log N})\). A natural open
problem is to remove the dependence on \(N\) constructively. In particular, can
one find, in polynomial time, a rainbow prefix decomposition with prefix error
\(O(k)\), or even \(O(\sqrt{k})\)? Alternatively, one could seek lower bounds
showing that such improvements are impossible.
\end{itemize}

\bibliographystyle{ACM-Reference-Format}
\bibliography{bibliography}

@article{ercan2026bound,
  title={A $(\log n)^{1/4}$ bound for the {Koml{\'o}s} problem},
  author={Ercan, Eren},
  journal={arXiv preprint arXiv:2609.08885},
  year={2026}
}

@article{karingula2026elementary,
  title={An elementary proof of the {Koml{\'o}s} conjecture},
  author={Karingula, Sankeerth Rao and Lovett, Shachar},
  journal={arXiv preprint arXiv:2609.20979},
  year={2026}
}

@article{guo2026vector,
  title={Vector balancing in polynomial time},
  author={Guo, Shengtao and Fang, Ethan X. and Lu, Junwei},
  journal={arXiv preprint arXiv:2609.23540},
  year={2026}
}

@inproceedings{jiang2022tighter,
  title={A Tighter Relation Between Hereditary Discrepancy and Determinant Lower Bound},
  author={Jiang, Haotian and Reis, Victor},
  booktitle={Symposium on Simplicity in Algorithms (SOSA)},
  pages={308--313},
  year={2022},
  organization={SIAM}
}

@book{reis2023vector,
  title={Vector Balancing and Integer Programming},
  author={Reis, Victor},
  year={2023},
  publisher={University of Washington}
}

@inproceedings{rothvoss2012entropy,
  title={The entropy rounding method in approximation algorithms},
  author={Rothvo{\ss}, Thomas},
  booktitle={Proceedings of the Twenty-Third Annual ACM-SIAM Symposium on Discrete Algorithms},
  pages={356--372},
  year={2012},
  organization={SIAM}
}

@inproceedings{bansal2019generalization,
  title={On a generalization of iterated and randomized rounding},
  author={Bansal, Nikhil},
  booktitle={Proceedings of the 51st Annual ACM SIGACT Symposium on Theory of Computing},
  pages={1125--1135},
  year={2019}
}

@article{gluskin1989extremal,
  title={Extremal properties of orthogonal parallelepipeds and their applications to the geometry of Banach spaces},
  author={Gluskin, Efim Davydovich},
  journal={Mathematics of the USSR-Sbornik},
  volume={64},
  number={1},
  pages={85--96},
  year={1989}
}

@article{olson1978balancing,
  title={Balancing families of sets},
  author={Olson, John E and Spencer, Joel H},
  journal={Journal of Combinatorial Theory, Series A},
  volume={25},
  number={1},
  pages={29--37},
  year={1978},
  publisher={Elsevier}
}

@article{baranyai1974factrization,
  title={On the factrization of the complete uniform hypergraphs},
  author={Baranyai, Zsolt},
  journal={Infinite and finite sets},
  year={1974}
}

@article{beck1981roth,
  title={Roth's estimate of the discrepancy of integer sequences is nearly sharp},
  author={Beck, J{\'o}zsef},
  journal={Combinatorica},
  volume={1},
  number={4},
  pages={319--325},
  year={1981},
  publisher={Springer}
}

@article{spencer2001discrepancy,
  title={The discrepancy of permutation families},
  author={Spencer, Joel H and Srinivasan, Aravind and Tetali, Prasad},
  journal={Unpublished manuscript},
  pages={21},
  year={2001}
}

@article{bansal2019algorithm,
  title={An algorithm for Koml{\'o}s conjecture matching Banaszczyk's bound},
  author={Bansal, Nikhil and Dadush, Daniel and Garg, Shashwat},
  journal={SIAM Journal on Computing},
  volume={48},
  number={2},
  pages={534--553},
  year={2019},
  publisher={SIAM}
}

@article{rothvoss2017constructive,
  title={Constructive discrepancy minimization for convex sets},
  author={Rothvoss, Thomas},
  journal={SIAM Journal on Computing},
  volume={46},
  number={1},
  pages={224--234},
  year={2017},
  publisher={SIAM}
}

@inproceedings{bar2005using,
  title={Using fractional primal-dual to schedule split intervals with demands},
  author={Bar-Yehuda, Reuven and Rawitz, Dror},
  booktitle={European Symposium on Algorithms},
  pages={714--725},
  year={2005},
  organization={Springer}
}

@incollection{bar2005local,
  title={The local ratio technique and its application to scheduling and resource allocation problems},
  author={Bar-Yehuda, Reuven and Bendel, Keren and Freund, Ari and Rawitz, Dror},
  booktitle={Graph Theory, Combinatorics and Algorithms: Interdisciplinary Applications},
  pages={107--143},
  year={2005},
  publisher={Springer}
}

@article{bar2001unified,
  title={A unified approach to approximating resource allocation and scheduling},
  author={Bar-Noy, Amotz and Bar-Yehuda, Reuven and Freund, Ari and Naor, Joseph and Schieber, Baruch},
  journal={Journal of the ACM (JACM)},
  volume={48},
  number={5},
  pages={1069--1090},
  year={2001},
  publisher={ACM New York, NY, USA}
}

@inproceedings{rohwedder2025cost,
  title={Cost Preserving Dependent Rounding for Allocation Problems},
  author={Rohwedder, Lars and Rouhani, Arman and Wennmann, Leo},
  booktitle={52nd EATCS International Colloquium on Automata, Languages, and Programming, ICALP 2025},
  pages={127},
  year={2025},
  organization={Schloss Dagstuhl-Leibniz-Zentrum fur Informatik GmbH, Dagstuhl Publishing}
}

@article{beck1981integer,
  title={``Integer-making'' theorems},
  author={Beck, J{\'o}zsef and Fiala, Tibor},
  journal={Discrete Applied Mathematics},
  volume={3},
  number={1},
  pages={1--8},
  year={1981},
  publisher={Elsevier}
}

@article{spencer1985six,
  title={Six standard deviations suffice},
  author={Spencer, Joel},
  journal={Transactions of the American mathematical society},
  volume={289},
  number={2},
  pages={679--706},
  year={1985}
}

@article{spencer1986balancing,
  title={Balancing vectors in the max norm},
  author={Spencer, Joel},
  journal={Combinatorica},
  volume={6},
  number={1},
  pages={55--65},
  year={1986},
  publisher={Springer}
}

@article{banaszczyk1998balancing,
  title={Balancing vectors and Gaussian measures of n-dimensional convex bodies},
  author={Banaszczyk, Wojciech},
  journal={Random Structures \& Algorithms},
  volume={12},
  number={4},
  pages={351--360},
  year={1998},
  publisher={Wiley Online Library}
}

@article{banaszczyk2012series,
  title={On series of signed vectors and their rearrangements},
  author={Banaszczyk, Wojciech},
  journal={Random Structures \& Algorithms},
  volume={40},
  number={3},
  pages={301--316},
  year={2012},
  publisher={Wiley Online Library}
}

@inproceedings{bansal2010constructive,
  title={Constructive algorithms for discrepancy minimization},
  author={Bansal, Nikhil},
  booktitle={2010 IEEE 51st Annual Symposium on Foundations of Computer Science},
  pages={3--10},
  year={2010},
  organization={IEEE}
}

@article{rothvoss2016better,
  title={Better bin packing approximations via discrepancy theory},
  author={Rothvoss, Thomas},
  journal={SIAM Journal on Computing},
  volume={45},
  number={3},
  pages={930--946},
  year={2016},
  publisher={SIAM}
}

@incollection{bansal2017approximation,
  title={Approximation-friendly discrepancy rounding},
  author={Bansal, Nikhil and Nagarajan, Viswanath},
  booktitle={A Journey Through Discrete Mathematics: A Tribute to Ji{\v{r}}{\'\i} Matou{\v{s}}ek},
  pages={89--114},
  year={2017},
  publisher={Springer}
}

@inproceedings{bansal2022flow,
  title={Flow time scheduling and prefix Beck-Fiala},
  author={Bansal, Nikhil and Rohwedder, Lars and Svensson, Ola},
  booktitle={Proceedings of the 54th Annual ACM SIGACT Symposium on Theory of Computing},
  pages={331--342},
  year={2022}
}

@inproceedings{bansal2026decoupling,
  title={Decoupling via affine spectral-independence: Beck-Fiala and Koml{\'o}s bounds beyond Banaszczyk},
  author={Bansal, Nikhil and Jiang, Haotian},
  booktitle={Proceedings of the 58th Annual ACM Symposium on Theory of Computing},
  pages={432--442},
  year={2026}
}

@article{baranyai1979edge,
  title={The edge-coloring of complete hypergraphs I},
  author={Baranyai, Zsolt},
  journal={Journal of Combinatorial Theory, Series B},
  volume={26},
  number={3},
  pages={276--294},
  year={1979},
  publisher={Elsevier}
}

@article{el2007existence,
  title={On the existence of a rainbow 1-factor in 1-factorizations of K},
  author={El-Zanati, SI and Plantholt, MJ and Sissokho, PA and Spence, LE},
  journal={Journal of Combinatorial Designs},
  volume={15},
  number={6},
  pages={487--490},
  year={2007},
  publisher={Wiley Online Library}
}

@article{li2007existence,
  title={On the existence of a rainbow 1-factor in proper coloring of K\_ $\{$rn$\}$\^{}$\{$(r)$\}$},
  author={Li, Xueliang and Xu, Zhixia},
  journal={arXiv preprint arXiv:0711.2847},
  year={2007}
}

@article{aharoni2009rainbow,
  title={Rainbow matchings in $ r $-partite $ r $-graphs},
  author={Aharoni, Ron and Berger, Eli},
  journal={The Electronic Journal of Combinatorics},
  pages={R119--R119},
  year={2009}
}

@article{gao2021full,
  title={Full rainbow matchings in graphs and hypergraphs},
  author={Gao, Pu and Ramadurai, Reshma and Wanless, Ian M and Wormald, Nick},
  journal={Combinatorics, Probability and Computing},
  volume={30},
  number={5},
  pages={762--780},
  year={2021},
  publisher={Cambridge University Press}
}

@inproceedings{bansal2022prefix,
  title={Prefix Discrepancy, Smoothed Analysis, and Combinatorial Vector Balancing},
  author={Bansal, Nikhil and Jiang, Haotian and Meka, Raghu and Singla, Sahil and Sinha, Makrand},
  booktitle={13th Innovations in Theoretical Computer Science Conference, ITCS 2022},
  pages={13},
  year={2022},
  organization={Schloss Dagstuhl-Leibniz-Zentrum fur Informatik GmbH, Dagstuhl Publishing}
}

@inproceedings{zuzic2023sparse,
  title={Sparse semi-oblivious routing: Few random paths suffice},
  author={Zuzic, Goran and Haeupler, Bernhard and Roeyskoe, Antti},
  booktitle={Proceedings of the 2023 ACM Symposium on Principles of Distributed Computing},
  pages={222--232},
  year={2023}
}

@article{hu2001survey,
  title={A survey on multi-net global routing for integrated circuits},
  author={Hu, Jiang and Sapatnekar, Sachin S},
  journal={Integration},
  volume={31},
  number={1},
  pages={1--49},
  year={2001},
  publisher={Elsevier}
}

@book{kahng2011vlsi,
  title={VLSI physical design: from graph partitioning to timing closure},
  author={Kahng, Andrew B and Lienig, Jens and Markov, Igor L and Hu, Jin},
  volume={312},
  year={2011},
  publisher={Springer}
}

@article{karp1987global,
  title={Global wire routing in two-dimensional arrays},
  author={Karp, Richard M and Leighton, Frank Thomson and Rivest, Ronald L and Thompson, Clark D and Vazirani, Umesh V and Vazirani, Vijay V},
  journal={Algorithmica},
  volume={2},
  number={1},
  pages={113--129},
  year={1987},
  publisher={Springer}
}

@inproceedings{bhatia1992improved,
  title={Improved algorithms for routing on two-dimensional grids},
  author={Bhatia, Dinesh and Leighton, Tom and Makedon, Fillia and Norton, Carolyn H},
  booktitle={International Workshop on Graph-Theoretic Concepts in Computer Science},
  pages={114--122},
  year={1992},
  organization={Springer}
}

@inproceedings{hajiaghayi2006semi,
  title={Semi-oblivious routing},
  author={Hajiaghayi, MohammadTaghi and Kleinberg, Robert and Leighton, Tom},
  booktitle={Proceedings of the eighteenth annual ACM symposium on Parallelism in algorithms and architectures},
  pages={234--234},
  year={2006}
}

@article{raghavan1987randomized,
  title={Randomized rounding: a technique for provably good algorithms and algorithmic proofs},
  author={Raghavan, Prabhakar and Tompson, Clark D},
  journal={Combinatorica},
  volume={7},
  number={4},
  pages={365--374},
  year={1987},
  publisher={Springer}
}

@article{skutella2002approximating,
  title={Approximating the single source unsplittable min-cost flow problem},
  author={Skutella, Martin},
  journal={Mathematical Programming},
  volume={91},
  number={3},
  pages={493--514},
  year={2002},
  publisher={Springer}
}

@inproceedings{traub2024single,
  title={Single-source unsplittable flows in planar graphs},
  author={Traub, Vera and Koch, Laura Vargas and Zenklusen, Rico},
  booktitle={Proceedings of the 2024 Annual ACM-SIAM Symposium on Discrete Algorithms (SODA)},
  pages={639--668},
  year={2024},
  organization={SIAM}
}

@inproceedings{majthoub2025integer,
  title={Integer and unsplittable multiflows in series-parallel digraphs},
  author={Majthoub Almoghrabi, Mohammed and Skutella, Martin and Warode, Philipp},
  booktitle={International Conference on Integer Programming and Combinatorial Optimization},
  pages={427--441},
  year={2025},
  organization={Springer}
}

@article{schrijver1999ring,
  title={The ring loading problem},
  author={Schrijver, Alexander and Seymour, Paul and Winkler, Peter},
  journal={SIAM review},
  volume={41},
  number={4},
  pages={777--791},
  year={1999},
  publisher={SIAM}
}

@article{skutella2016note,
  title={A note on the ring loading problem},
  author={Skutella, Martin},
  journal={SIAM Journal on Discrete Mathematics},
  volume={30},
  number={1},
  pages={327--342},
  year={2016},
  publisher={SIAM}
}

@article{daeubel2022improved,
  title={An improved upper bound for the ring loading problem},
author={D{\"a}ubel, Karl},
  journal={SIAM Journal on Discrete Mathematics},
  volume={36},
  number={2},
  pages={867--887},
  year={2022},
  publisher={SIAM}
}

@inproceedings{aleman2025unsplittable,
  title={Unsplittable multicommodity flows in outerplanar graphs},
  author={Alem{\'a}n-Espinosa, David and Kumar, Nikhil},
  booktitle={International Conference on Integer Programming and Combinatorial Optimization},
  pages={385--399},
  year={2025},
  organization={Springer}
}

@inproceedings{aleman2026pinning,
  title={Pinning on Tight Cuts: Improved Algorithm and Bounds for Unsplittable Multicommodity Flows in Outerplanar Graphs},
  author={Alem{\'a}n Espinosa, David and Schlomberg, Niklas},
  booktitle={53rd International Colloquium on Automata, Languages, and Programming (ICALP 2026)},
  pages={11--1},
  year={2026},
  organization={Schloss Dagstuhl--Leibniz-Zentrum f{\"u}r Informatik}
}

@inproceedings{hoberg2017logarithmic,
  title={A logarithmic additive integrality gap for bin packing},
  author={Hoberg, Rebecca and Rothvoss, Thomas},
  booktitle={Proceedings of the Twenty-Eighth Annual ACM-SIAM Symposium on Discrete Algorithms},
  pages={2616--2625},
  year={2017},
  organization={SIAM}
}

@article{lovett2015constructive,
  title={Constructive discrepancy minimization by walking on the edges},
  author={Lovett, Shachar and Meka, Raghu},
  journal={SIAM Journal on Computing},
  volume={44},
  number={5},
  pages={1573--1582},
  year={2015},
  publisher={SIAM}
}

@inproceedings{doerr2005roundings,
  title={Roundings respecting hard constraints},
  author={Doerr, Benjamin},
  booktitle={Annual Symposium on Theoretical Aspects of Computer Science},
  pages={617--628},
  year={2005},
  organization={Springer}
}

@inproceedings{grigoriev2005unrelated,
  title={Unrelated parallel machine scheduling with resource dependent processing times},
  author={Grigoriev, Alexander and Sviridenko, Maxim and Uetz, Marc},
  booktitle={International Conference on Integer Programming and Combinatorial Optimization},
  pages={182--195},
  year={2005},
  organization={Springer}
}

@article{fiala1983algorithm,
  title={An algorithm for the open-shop problem},
  author={Fiala, Tibor},
  journal={Mathematics of Operations Research},
  volume={8},
  number={1},
  pages={100--109},
  year={1983},
  publisher={INFORMS}
}

@inproceedings{bansal2014better,
  title={Better algorithms and hardness for broadcast scheduling via a discrepancy approach},
  author={Bansal, Nikhil and Charikar, Moses and Krishnaswamy, Ravishankar and Li, Shi},
  booktitle={Proceedings of the twenty-fifth annual ACM-SIAM symposium on Discrete algorithms},
  pages={55--71},
  year={2014},
  organization={SIAM}
}

@inproceedings{fan2024multi,
  title={The multi-commodity flow problem: Double dantzig-wolfe decomposition},
  author={Fan, Zhang and Jiazheng, Wang and Lacroix, Mathieu and Calvo, Roberto Wolfler and Magnouche, Youcef and Martin, Sebastien},
  booktitle={2024 10th International Conference on Control, Decision and Information Technologies (CoDIT)},
  pages={1171--1176},
  year={2024},
  organization={IEEE}
}

@article{chakrabarti2007approximation,
  title={Approximation algorithms for the unsplittable flow problem},
  author={Chakrabarti, Amit and Chekuri, Chandra and Gupta, Anupam and Kumar, Amit},
  journal={Algorithmica},
  volume={47},
  number={1},
  pages={53--78},
  year={2007},
  publisher={Springer}
}

@inproceedings{harris2013moser,
  title={The Moser-Tardos framework with partial resampling},
  author={Harris, David G and Srinivasan, Aravind},
  booktitle={2013 IEEE 54th annual symposium on foundations of computer science},
  pages={469--478},
  year={2013},
  organization={IEEE}
}

@inproceedings{newman2012beck,
  title={Beck's three permutations conjecture: A counterexample and some consequences},
  author={Newman, Alantha and Neiman, Ofer and Nikolov, Aleksandar},
  booktitle={2012 IEEE 53rd Annual Symposium on Foundations of Computer Science},
  pages={253--262},
  year={2012},
  organization={IEEE}
}

@article{kumar2009unified,
  title={A unified approach to scheduling on unrelated parallel machines},
  author={Kumar, VS Anil and Marathe, Madhav V and Parthasarathy, Srinivasan and Srinivasan, Aravind},
  journal={Journal of the ACM (JACM)},
  volume={56},
  number={5},
  pages={1--31},
  year={2009},
  publisher={ACM New York, NY, USA}
}

@article{lovasz1986discrepancy,
  title={Discrepancy of set-systems and matrices},
  author={Lov{\'a}sz, L{\'a}szl{\'o} and Spencer, Joel and Vesztergombi, Katalin},
  journal={European Journal of Combinatorics},
  volume={7},
  number={2},
  pages={151--160},
  year={1986},
  publisher={Elsevier}
}

@article{biedl2002balanced,
  title={Balanced k-colorings},
  author={Biedl, Therese C and {\v{C}}enek, Eowyn and Chan, Timothy M and Demaine, Erik D and Demaine, Martin L and Fleischer, Rudolf and Wang, Ming-Wei},
  journal={Discrete Mathematics},
  volume={254},
  number={1-3},
  pages={19--32},
  year={2002},
  publisher={Elsevier}
}

@article{gandhi2006dependent,
  title={Dependent rounding and its applications to approximation algorithms},
  author={Gandhi, Rajiv and Khuller, Samir and Parthasarathy, Srinivasan and Srinivasan, Aravind},
  journal={Journal of the ACM (JACM)},
  volume={53},
  number={3},
  pages={324--360},
  year={2006},
  publisher={ACM New York, NY, USA}
}

@inproceedings{leighton1998multicommodity,
  title={Multicommodity flow and circuit switching},
  author={Leighton, Tom and Rao, Satish and Srinivasan, Aravind},
  booktitle={Proceedings of the Thirty-First Hawaii International Conference on System Sciences},
  volume={7},
  pages={459--465},
  year={1998},
  organization={IEEE}
}

@article{eisenbrand2013bin,
  title={Bin packing via discrepancy of permutations},
  author={Eisenbrand, Friedrich and P{\'a}lv{\"o}lgyi, D{\"o}m{\"o}t{\"o}r and Rothvo{\ss}, Thomas},
  journal={ACM Transactions on Algorithms (TALG)},
  volume={9},
  number={3},
  pages={1--15},
  year={2013},
  publisher={ACM New York, NY, USA}
}

@inproceedings{keevash2018hypergraph,
  title={Hypergraph matchings and designs},
  author={Keevash, P},
  booktitle={Proceedings of the International Congress of Mathematicians (ICM 2018)},
  volume={3},
  year={2019},
  organization={World Scientific Publishing}
}

@article{glock2021extremal,
  title={Extremal aspects of graph and hypergraph decomposition problems},
  author={Glock, Stefan and K{\"u}hn, Daniela and Osthus, Deryk},
  journal={arXiv preprint arXiv:2008.00926},
  year={2020}
}

@article{rodl1985packing,
  title={On a packing and covering problem},
  author={R{\"o}dl, Vojt{\v{e}}ch},
  journal={European Journal of Combinatorics},
  volume={6},
  number={1},
  pages={69--78},
  year={1985},
  publisher={Elsevier}
}

@article{pippenger1989asymptotic,
  title={Asymptotic behavior of the chromatic index for hypergraphs},
  author={Pippenger, Nicholas and Spencer, Joel},
  journal={Journal of combinatorial theory, Series A},
  volume={51},
  number={1},
  pages={24--42},
  year={1989},
  publisher={Elsevier}
}

@article{keevash2016existence,
  title={The existence of designs},
  author={Keevash, Peter},
  journal={arXiv preprint arXiv:1401.3665},
  year={2014}
}

@book{glock2023existence,
  title={The existence of designs via iterative absorption: hypergraph F-designs for arbitrary F},
  author={Glock, Stefan and K{\"u}hn, Daniela and Lo, Allan and Osthus, Deryk},
  volume={284},
  number={1406},
  year={2023},
  publisher={American Mathematical Society}
}

@article{keevash2018rainbow,
  title={Rainbow matchings in properly colored multigraphs},
  author={Keevash, Peter and Yepremyan, Liana},
  journal={SIAM journal on Discrete Mathematics},
  volume={32},
  number={3},
  pages={1577--1584},
  year={2018},
  publisher={SIAM}
}

@article{coulson2020rainbow,
title = {Rainbow factors in hypergraphs},
journal = {Journal of Combinatorial Theory, Series A},
volume = {172},
pages = {105184},
year = {2020},
issn = {0097-3165},
doi = {https://doi.org/10.1016/j.jcta.2019.105184},
author = {Matthew Coulson and Peter Keevash and Guillem Perarnau and Liana Yepremyan},
}

\appendix

\section{Proof of Corollary~\ref{cor:discHardConstraints}}
\label{app:discHardConstraints}

For \(B\in\mathbb R^{p\times q}\), its \emph{determinant lower bound} is
\[
    \detlb(B):=
    \max_{\substack{I\subseteq[p],\,J\subseteq[q]\\|I|=|J|\ge1}}
    |\det(B_{I,J})|^{1/|I|},
\]
with \(\detlb(B):=0\) if \(B\) has no columns.
We use the following algorithm as our discrepancy oracle.

\begin{theorem}[Jiang--Reis~\cite{jiang2022tighter}]
\label{th:JiangReisOracle}
There is a randomized polynomial-time algorithm that, given
\(B\in\mathbb R^{p\times q}\), finds, with high probability,
a coloring \(\chi\in\{-1,1\}^q\) satisfying
\[
    \|B\chi\|_\infty
    \le
    O\!\left(\sqrt{\log(2+p)\log(2+q)}\right)\detlb(B).
\]
\end{theorem}

\begin{proof}[Proof of Corollary~\ref{cor:discHardConstraints}]
Let \(D:=\herdisc(A)\) and fix \(Q\in\mathcal Q(A)\).
Every column of \(Q\) is the difference of two columns of \(A\),
and these pairs are disjoint. By multilinearity, the determinant
of any \(s\times s\) submatrix \(M\) of \(Q\) is a signed sum
of \(2^s\) determinants of \(s\times s\) submatrices of \(A\). Thus
\[
    |\det M|\le 2^s\detlb(A)^s,
    \qquad
    \detlb(Q)\le 2\detlb(A)\le 4D,
\]
where the last inequality is the determinant lower bound of
Lov\'asz, Spencer, and Vesztergombi~\cite{lovasz1986discrepancy}.

Since \(Q\) has \(m\) rows and at most \(m\) columns,
Theorem~\ref{th:JiangReisOracle} provides a discrepancy oracle with
\[
    \beta=O\!\left(\log(2+m)D\right).
\]
Applying Theorem~\ref{th:discHardConstraints} gives
\(y\in\{0,1\}^n\) such that
\[
    Hy=Hx,
    \qquad
    \langle c,y-x\rangle\le0,
    \qquad
    \|A(x-y)\|_\infty
    \le \beta+\frac{2}{n}\|A\|_\infty
    =O\!\left(\log(2+m)D\right),
\]
where we used \(\|A\|_\infty\le D\), which follows by considering
single-column restrictions of \(A\).
\end{proof}

\section{Rounding with block-TU hard constraints}
\label{app:block-TU-hard-constraints}

Throughout this section, let
\[
    H=\operatorname{diag}(H^1,\ldots,H^s)\in\mathbb Z^{k\times n},
\]
where each block is totally unimodular and has at most \(r\) rows.
Recall that a matrix is \emph{totally unimodular} if every square
minor is \(0\) or \(\pm1\).

A \emph{circuit vector} of \(H\) is a nonzero integer vector in
\(\ker(H)\) with inclusion-minimal support and relatively prime
entries. By total unimodularity, its nonzero entries are \(\pm1\).
Moreover, each circuit is contained in one block, so its support
has size at most \(r+1\).
For \(A\in\mathbb R^{m\times n}\), define
\[
    \mathcal Q_H(A):=
    \left\{
        [Ag^1,\ldots,Ag^p]:
        \begin{array}{l}
            0\le p\le m,\quad g^1,\ldots,g^p
            \text{ are circuit vectors of }H,\\
            \text{with pairwise disjoint supports}
        \end{array}
    \right\}.
\]
When every block is an all-ones row, these circuits are precisely
the pairwise exchanges used in
Theorem~\ref{th:discHardConstraints}.

\begin{lemma}
\label{lem:TU-parity-decomposition}
If \(S\subseteq[n]\) satisfies
\(H_{\ast S}\mathbf1\in(2\mathbb Z)^k\), then \(S\) can be
partitioned into supports of circuit vectors of \(H\).
Such a partition can be found in polynomial time.
\end{lemma}

\begin{proof}
For a totally unimodular matrix, column independence over
\(\mathbb R\) and over \(\mathbb F_2\) is equivalent, since every
nonzero minor is \(\pm1\).
Thus, if \(S\ne\varnothing\), its even column sum implies that it
contains the support \(C\) of a circuit vector \(g\).
Since \(g_j=\pm1\) on \(C\) and \(Hg=0\), the column sum over
\(C\) is also even. Remove \(C\) and repeat.
Each circuit can be found by Gaussian elimination.
\end{proof}

\begin{theorem}
\label{th:block-TU-hard-constraints}
Let \(A\in\mathbb R^{m\times n}\), let \(H\) be as above, and let
\(c\in\mathbb R^n\). Suppose that an oracle returns, for every
\(Q\in\mathcal Q_H(A)\), a coloring \(\chi\) satisfying
\(\|Q\chi\|_\infty\le\beta\), where \(\beta\ge0\).
Then, for every \(x\in[0,1]^n\) with \(Hx\in\mathbb Z^k\),
one can find \(y\in\{0,1\}^n\) such that
\[
    Hy=Hx,
    \qquad
    \langle c,y-x\rangle\le0,
    \qquad
    \|A(x-y)\|_\infty
    \le \beta+\frac{1}{n}\|A\|_\infty.
\]
The algorithm uses at most \(\lceil2\log_2 n\rceil\) oracle calls
and otherwise runs in polynomial time.
\end{theorem}

\begin{proof}
Let \(x'\) be a vertex minimizing \(\langle c,z\rangle\) over $\{z\in[0,1]^n:Az=Ax,\ Hz=Hx\},$
and let \(F:=\{j:0<x'_j<1\}\).
The columns of \(\binom{A}{H}\) indexed by \(F\) are independent,
hence
\[
    \dim\ker(H_{\ast F})
    =|F|-\operatorname{rank}(H_{\ast F})\le m.
\]

Set \(L:=\lceil2\log_2 n\rceil\) and \(t:=2^L\).
The polytope $\{v:Hv=tHx,\ \lfloor tx'\rfloor\le v\le\lceil tx'\rceil\}$
contains \(tx'\) and is integral by total unimodularity.
Choose an integral point \(v\) minimizing \(\langle c,v\rangle\)
and set \(w:=v/t\). Then
\[
    Hw=Hx,\qquad
    \langle c,w-x'\rangle\le0,\qquad
    \|w-x'\|_1\le n/t\le1/n.
\]
Also, \(w_j=x'_j\) outside \(F\).

Starting with \(u=w\), process the levels
\(\ell=L,\ldots,1\). At level \(\ell\), let
$S:=\{j\in F:2^\ell u_j\text{ is odd}\}.$
Since \(Hu=Hx\in\mathbb Z^k\),
\[
    H_{\ast S}\mathbf1
    \equiv 2^\ell Hu
    \equiv0\pmod2.
\]
By Lemma~\ref{lem:TU-parity-decomposition}, circuit vectors
\(g^1,\ldots,g^p\) have supports partitioning \(S\).
Their restrictions to \(F\) are linearly independent and belong
to \(\ker(H_{\ast F})\), so \(p\le m\).

Apply the oracle to \(Q=[Ag^1,\ldots,Ag^p]\in\mathcal Q_H(A)\).
Reverse the returned coloring \(\chi\), if necessary, so that
\(\sum_a\chi_a\langle c,g^a\rangle\le0\), and update
\[
    u\leftarrow u+2^{-\ell}\sum_{a=1}^p\chi_ag^a.
\]
Each odd numerator changes by \(\pm1\), so the updated vector
belongs to \([0,1]^n\cap2^{-(\ell-1)}\mathbb Z^n\).
The update preserves \(Hu\), does not increase the cost, and
changes \(Au\) by at most \(2^{-\ell}\beta\) in the infinity norm.

The final vector \(y\) is integral, satisfies \(Hy=Hx\), and
has cost at most that of \(x\). Since \(Ax'=Ax\),
\[
    \|A(x-y)\|_\infty
    \le
    \frac1n\|A\|_\infty+\sum_{\ell=1}^L2^{-\ell}\beta
    \le
    \beta+\frac1n\|A\|_\infty.
\]
There is at most one oracle call per level.
\end{proof}

\begin{corollary}
\label{cor:block-TU-herdisc}
Let \(A\in\mathbb R^{m\times n}\), let \(H\) be as above, and
let \(c\in\mathbb R^n\).
For every \(x\in[0,1]^n\) with \(Hx\in\mathbb Z^k\),
a randomized polynomial-time algorithm finds, with high probability,
\(y\in\{0,1\}^n\) such that
\[
    Hy=Hx,
    \qquad
    \langle c,y-x\rangle\le0,
    \qquad
    \|A(x-y)\|_\infty
    \le O\!\left((r+1)\log(2+m)\herdisc(A)\right).
\]
\end{corollary}

\begin{proof}
Fix \(Q\in\mathcal Q_H(A)\).
By multilinearity and disjointness of the circuit supports,
each \(d\times d\) minor of \(Q\) is a signed sum of at most
\((r+1)^d\) minors of \(A\). Therefore
\[
    \detlb(Q)\le(r+1)\detlb(A)
    \le2(r+1)\herdisc(A).
\]
Since \(Q\) has at most \(m\) columns, the Jiang--Reis oracle
used in the proof of Corollary~\ref{cor:discHardConstraints}
gives
\[
    \beta=O\!\left((r+1)\log(2+m)\herdisc(A)\right).
\]
Apply Theorem~\ref{th:block-TU-hard-constraints} and use
\(\|A\|_\infty\le\herdisc(A)\).
\end{proof}

\end{document}